\documentclass{amsart}
\usepackage{amsmath}%
\usepackage{amsfonts}%
\usepackage{amssymb}%
\usepackage{graphicx}
\newtheorem{theorem}{Theorem}
\theoremstyle{plain}

\newtheorem{corollary}{Corollary}

\newtheorem{definition}{Definition}

\newtheorem{lemma}{Lemma}

\newtheorem{preremark}{Remark}
\newenvironment{remark}{\begin{preremark}\rm}{\end{preremark}}

\numberwithin{equation}{section}

\def\R{\mathbb{R}}
\def\N{\mathbb{N}}
\def\K{\mathcal{K}}

\begin{document}
\title[Global Classical Solutions Vlasov-Darwin Small Data]{Global Classical Solutions of the relativistic Vlasov-Darwin system with small Cauchy Data: the generalized variables approach}
%\title[Classical Solvability for Vlasov-Maxwell]{Global Classical solvability of the 3D Relativistic Vlasov-Maxwell is guaranteed if the spatial density remains bounded}
\author{Reinel Sospedra-Alfonso}
\author{Martial Agueh}
\author{Reinhard Illner}
\address%[A. One and A. Two]
{Institute of Applied Mathematics,\newline%
\indent University of British Columbia, Mathematics Road, Vancouver BC, Canada V6T 1Z2.}%
\email[R. Sospedra-Alfonso]{sospedra@chem.ubc.ca}%
%\urladdr{http://www.authorone.uni-aone.de}
\address{Department of Mathematics and Statistics,\newline%
\indent University of Victoria, PO BOX 3045 STN CSC, Victoria BC, Canada V8W 3P4.}%
\email[M. Agueh]{agueh@math.uvic.ca}%
%\curraddr[A.~Two]{Author Two current address, line 1\newline%
%\indent Author Two current address, line 2}%
%\urladdr{http://www.authortwo.uni-atwo.hu}
\email[R. Illner]{rillner@math.uvic.ca}%
%\curraddr[A.~Two]{Author Two current address, line 1\newline%
%\indent Author Two current address, line 2}%
%\urladdr{http://www.authortwo.uni-atwo.hu}

%\thanks{Thanks for Author One.}
%\thanks{Thanks for Author Two.}
%\thanks{This paper is in final form and no version of it will be submitted for
%publication elsewhere.}

%\thanks{Martial Agueh and Reinhard Illner are supported by Discovery grants from the Natural Science and Engineering Research Council of Canada (NSERC)}

\date{May 12, 2012}
%\subjclass{Primary 05C38, 15A15; Secondary 05A15, 15A18} %
%\dedicatory{Dedicated to Professor XY on the occasion of his seventieth birthday.}

\begin{abstract}
  We show that a smooth, small enough Cauchy datum launches a unique classical solution of the relativistic Vlasov-Darwin (RVD) system globally in time. A similar result is claimed in \cite{Seehafer} following the work in \cite{Pallard}. Our proof does not require estimates derived from the conservation of the total energy, nor those previously given on the transversal component of the electric field. These estimates are crucial in the references cited above. Instead, we exploit the formulation of the RVD system in terms of the \textit{generalized} space and momentum variables. By doing so, we produce a simple a-priori estimate on the transversal component of the electric field. We widen the functional space required for the Cauchy datum to extend the solution globally in time, and we improve decay estimates given in \cite{Seehafer} on the electromagnetic field and its space derivatives. Our method extends the constructive proof presented in \cite{Rein} to solve the Cauchy problem for the Vlasov-Poisson system with a small initial datum.
\end{abstract}
\maketitle

\section{Introduction}
\label{intro}

The relativistic Vlasov-Darwin (RVD) system can be obtained from the Vlasov-Maxwell system by neglecting the transversal component of the displacement current in the Maxwell-Amp{\`e}re equation. More precisely, consider an ensemble of single species charged particles interacting through the self-induced electromagnetic field. Let $f(t,x,p)$ denote the number of particles per unit volume of the phase-space at a time $t\in\left]0,\infty\right[$, where $x\in\mathbb{R}^3$ is position and $p\in\mathbb{R}^3$ denotes momentum. In the regime in which collisions among the particles can be neglected, the time evolution of the distribution function $f$ is given by the Vlasov equation
\begin{equation}
\label{Vlasov For Maxwell}
\partial_tf+v\cdot\nabla_{x}f+\left(E+c^{-1}v\times B\right)\cdot\nabla_{p}f=0,\quad v=\frac{p}{\sqrt{1+c^{-2}\left|p\right|^2}}, 
\end{equation}
where $v$ is the relativistic velocity and $c$ the speed of light. Here the mass and charge of the particles have been set to one. $E=E(t,x)$ and $B=B(t,x)$ denote the self-induced electric and magnetic fields, given by the Maxwell equations
\begin{eqnarray}
\label{Maxwell 1}
\nabla\times B-c^{-1}\partial_tE & = & 4\pi c^{-1}j,\quad\quad \nabla\cdot B\quad  =\quad  0,\\
\label{Maxwell 2}
\nabla\times E +c^{-1}\partial_tB & = & 0, \quad\quad\quad\quad\ \, \nabla\cdot E\quad  =\quad  4\pi\rho.
\end{eqnarray}
The Vlasov and Maxwell equations are then coupled via the charge and current densities
\begin{equation}
\label{Density and Current}
  \rho=\int_{\mathbb{R}^3}fdp \quad \hbox{and} \quad j=\int_{\mathbb{R}^3}vfdp.
\end{equation}

Equations (\ref{Vlasov For Maxwell})-(\ref{Density and Current}) are known as the relativistic Vlasov-Maxwell (RVM) system, which is essential in the study of dilute hot plasmas. Details and an abundant bibliography on this system can be found, for instance, in \cite{GlasseyBook}. 

We further decompose the electric field into $E=E_L+E_T$, where the longitudinal $E_L$ and transversal $E_T$ components of the electric field satisfy, respectively
\begin{equation}
\label{Components Electric Field}
  \nabla\times E_L=0 \quad \hbox{and} \quad \nabla\cdot E_T=0.
\end{equation}
If we now neglect the transversal component of the displacement current $\partial_tE_T$ in the evolution equation (\ref{Maxwell 1}) -the so-called Maxwell-Amp{\`e}re equation-, then the RVM system reduces to 
\begin{equation}
\label{Vlasov For Darwin}
\partial_tf+v\cdot\nabla_{x}f+\left(E_L+E_T+c^{-1}v\times B\right)\cdot\nabla_{p}f=0, \quad v=\frac{p}{\sqrt{1+c^{-2}\left|p\right|^2}},
\end{equation}
coupled with
\begin{eqnarray}
\label{Darwin 1}
\nabla\times B-c^{-1}\partial_tE_L & = & 4\pi c^{-1}j,\quad \nabla\cdot B\quad  =\quad  0,\\
\label{Darwin 2}
\nabla\times E_T +c^{-1}\partial_tB & = & 0, \quad\quad\quad \nabla\cdot E_L\quad =\quad  4\pi\rho,
\end{eqnarray}
by means of (\ref{Density and Current}). Equations (\ref{Density and Current})-(\ref{Darwin 2}) are the RVD system. From the physical point of view, the Darwin approximation is valid when the evolution of the electromagnetic field is `slower' than the speed of light.

In this paper we are concerned with the Cauchy problem for (\ref{Density and Current})-(\ref{Darwin 2}). Global existence of weak solutions was shown in \cite{Benachour} for small initial data. The smallness assumption was later removed in \cite{Pallard}, where the existence and uniqueness of local in time classical solutions was also proved. In \cite{Seehafer}, it is shown that solutions having the same regularity as the initial data (which is not the case in \cite{Pallard}), can be extended globally in time provided the initial data is small. At the present time, the existence of global in time classical solutions for arbitrary data remains unsolved. Here, we provide a constructive and somewhat simplified proof to the local in time existence and uniqueness result for classical solutions of the RVD system, and we show that the solutions can be extended for all times if the initial data are sufficiently small.

The main difficulty when dealing with the RVD system has been to find an a-priori estimate on the transversal component of the electric field $E_T$. In contrast to the RVM system, the component $E_T$ does not contribute to the energy of the electromagnetic field, and thus the law for the conservation of the total energy does not provide any control on the $L^2$-norm of $E_T$. Indeed, the total energy of the RVD system reads
$$ \int_{\mathbb{R}^3}\int_{\mathbb{R}^3}c^2\sqrt{1+c^{-2}\left|p\right|^2}f(t,x,p)dpdx+\frac{1}{8\pi}\int_{\mathbb{R}^3}\left[\left|E_L(t,x)\right|^2+\left|B(t,x)\right|^2\right]dx.
$$
Hence, by virtue of the underlying elliptic structure of the Darwin equations, duality type arguments and variational methods have previously been used to estimate $E_T$. Here, we take advantage of the formulation of the RVD system in terms of the generalized variables, defined later on, and we produce an $L^2$-bound on $\rho^{1/2}E_T$ instead. This estimate is at the core of our results, and is given in Lemma \ref{L2 Estimate Time Derivative A}. It is remarkable that by pursuing such an estimate we have obtained, `almost for free', an $L^2$-bound on $\partial_xE_T$ as well. In contrast to the results cited above, the law for the conservation of the total energy is not used in our proofs at all.

The structure of the paper is as follows. In Section \ref{The Potential Representation Section}, we present the scalar and vector potentials and introduce the generalized position and momentum variables. We then recast the Vlasov and Darwin equations in terms of the new variables, and treat them as uncoupled linear equations. A representation for the Darwin vector potential is given, and some standard a-priori bounds are obtained as well. Then, in Section \ref{The RVD System Section} we couple both Vlasov and Darwin equations and introduce the RVD system in terms of the potentials. The estimates on the transversal component of the electric field and its space derivative are produced in Subsection \ref{Subsection Time Derivatives A}. Finally, we study the Cauchy problem for the RVD system in Section \ref{Section Small Data Solutions}. First, we produce the local in time existence result in Subsection \ref{Local Solutions Section} and then, in Subsection \ref{Global Solutions Section}, we extend local solutions globally in time under the smallness assumption on the Cauchy data. We conclude with an Appendix. 

We remark that the RVD is actually an hybrid system, since we are considering relativistic charged particles whose interaction with the electromagnetic field they induce is an order-$(v/c)^2$ approximation \cite{Jackson,Morrison}. Yet, the RVD system is interesting in its own right, in particular for numerical simulations, since it contains an underlying elliptic feature while preserving a fully coupled magnetic field. This is in contrast to the more involved RVM system, whose hyperbolic structure yields both analytical and numerical challenges. Also, the tools used here are likely to be adapted to the `proper' physical system, which is (\ref{Density and Current})-(\ref{Darwin 2}) with $v=p\left(1-c^{-2}p^2/2\right)$ instead.

The following notations will be used in  the paper. As usual, $C^{k,\alpha}(X;Y)$ denotes the space of functions $f: X \rightarrow Y$ of class $C^k$ whose $k^{th}$ derivatives are H\"older continuous with exponent $\alpha \in (0,1)$. $C^k_0(X;Y)$, resp. $C^k_b(X;Y)$, are the spaces of $C^k(X;Y)$-functions with compact support, resp. bounded. $W^{1,\infty}(X;Y)$ stands for the Sobolev space of $L^\infty(X;Y)$-functions whose weak first order partial derivatives belong to $L^\infty(X;Y)$. If $I$ is an interval in $\mathbb{R}$, then by $g\in C^1\left(I, C^k(X); Y\right)$, we mean that $g: I\times X \rightarrow Y$, $g=g(t,x)$, and for all $t\in I$, $g(t)\in C^k(X;Y)$ and the function $t \mapsto g(t)\in C^k(X;Y)$ is of class $C^1$ on $I$. For such a function, we sometimes write  (by abuse of notations) $g\in C^k(X;Y)$ to mean that $g(t)\in C^k(X;Y)$ for all $t\in I$. Similarly, the norm of $g(t)$, say the $L^q$-norm $\|g(t)\|_{L^q_x}$, will sometimes be denoted by $\|g\|_{L^q_x}$. All other notations  in the paper are standard, and the constants may change values from line to line.

%%%%%%%%%%%%%%%%%%%%%%%%%%%%%%%%%%%%%%%%%%%%%%%%%%%%%%%%%% 1 %%%%%%%%%%%%%%%%%%%%%%%%%%%%%%%%%%%%%%%%
%%%%%%%%%%%%%%%%%%%%%%%%%%%%%%%%%%%%%%%%%%%%%%%%%%%%%%%%%%%%%%%%%%%%%%%%%%%%%%%%%%%%%%%%%%%%%%%%%%%%%
\section{The Potential Representation}
\label{The Potential Representation Section}

From classical electrodynamics it is known that an electromagnetic field $(E,B):\mathbb{R}\times\mathbb{R}^3\rightarrow\mathbb{R}^3\times\mathbb{R}^3$ that is a smooth solution of the Maxwell equations (\ref{Maxwell 1})-(\ref{Maxwell 2}) can be represented by a set of potentials $(\Phi,A):(0,\infty)\times\mathbb{R}^3\rightarrow\mathbb{R}\times\mathbb{R}^3$ according to the expressions
\begin{eqnarray}
\label{Electric Field}
E(t,x) & = & -\nabla\Phi(t,x)-c^{-1}\partial_tA(t,x),\\
\label{Magnetic Field}
B(t,x) & = & \nabla\times A(t,x).
\end{eqnarray}
These relations can easily be obtained from the two homogeneous Maxwell equations in (\ref{Maxwell 1})-(\ref{Maxwell 2}). In particular (\ref{Magnetic Field}) follows from the vanishing divergence of the magnetic field, while (\ref{Electric Field}) follows after inserting (\ref{Magnetic Field}) into the remaining homogeneous equation. Since for any smooth scalar function $\Lambda$ we have the identity $\nabla\times\nabla\Lambda\equiv0$, it is clear that such potentials are not uniquely determined. We may find another pair given by 
\begin{equation}
\label{Gauge Transformation}
\left(\Phi',A'\right)=\left(\Phi-c^{-1}\partial_t\Lambda, A+\nabla\Lambda\right)
\end{equation}
which also satisfies (\ref{Electric Field})-(\ref{Magnetic Field}). The two sets of potentials are fully equivalent, in the sense that they produce the same electric and magnetic fields. 

This lack of uniqueness allows to impose a condition on the potentials that ultimately determines their dynamical equations. Even after doing so, some arbitrariness remains that can be avoided by imposing an additional restriction on $\Lambda$. The resulting restricted class is called a gauge, and all potentials within this class satisfy the same gauge condition. Commonly, the Lorentz gauge condition 
\begin{equation}
\label{Lorentz Gauge}
\nabla\cdot A+c^{-1}\partial_t\Phi=0,
\end{equation}
or the Coulomb gauge condition
\begin{equation}
\label{Coulomb Gauge}
\nabla\cdot A=0
\end{equation}
is used. The former is relativistically covariant and leads to a class of scalar and vector potentials that satisfy wave equations. This is a natural choice when dealing with the RVM system. It was used in \cite{Bouchut} to study the smoothing effect resulting from a coupling of a wave and transport equations. It was also used in \cite{BouchutGS1} to produce an alternative proof of the celebrated result by Glassey and Strauss on the RVM system \cite{GS1}. On the other hand, the Coulomb gauge condition leads to scalar and vector potentials that satisfy a Poisson and a wave equation, respectively. As we shall see in Subsection \ref{Darwin Potentials} below, this is the correct choice to introduce the potential representation of the RVD system. In a way, both the Lorentz and Coulomb gauges can be seen as limit cases of a more general class known as the velocity gauge, in which the scalar potential propagates with an arbitrary speed \cite{Jackson}.

%%%%%%%%%%%%%%%%%%%%%%%%%%%%%%%%%%%%%%%%%%%%%%%%%%%%%%%%%% 2 %%%%%%%%%%%%%%%%%%%%%%%%%%%%%%%%%%%%%%%%%%%%%%
%%%%%%%%%%%%%%%%%%%%%%%%%%%%%%%%%%%%%%%%%%%%%%%%%%%%%%%%%%%%%%%%%%%%%%%%%%%%%%%%%%%%%%%%%%%%%%%%%%%%%%%%%%%

\subsection{The Vlasov Equation}
\label{The Vlasov Equation Section}

We now introduce the generalized variables, which permit to rewrite the Vlasov equation (\ref{Vlasov For Maxwell}) in terms of the scalar and vector potentials in a very convenient way. The resulting transport equation is shown to be determined by an incompressible vector field \textit{irrespective of the gauge chosen}. Thus, we can count on the usual a-priori estimates on the distribution function -see Lemma \ref{Volume Preserving Lemma} below- no matter which gauge we decide to work in. 

To start with, let $I\subset[0,\infty[$ such that $0\in I$. Assume that the pair $(\Phi,A)\in C^1(I,C^2(\mathbb{R}^3);\mathbb{R}\times\mathbb{R}^3)$ is given, and so in view of (\ref{Electric Field})-(\ref{Magnetic Field}) the electromagnetic field is given as well. Denote $z:=(x,p)$. Then, by virtue of (\ref{Electric Field})-(\ref{Magnetic Field}), the characteristic system associated to the Vlasov equation (\ref{Vlasov For Maxwell}) reads
\begin{eqnarray}
\label{Characteristics X}
\dot{X}(s) & = & v(P(s)),\\
\label{Characteristics P}
\dot{P}(s) & = & \left[-\nabla\Phi-c^{-1}\partial_tA+c^{-1}v\times\left(\nabla\times A\right)\right](s,X(s),P(s)),
\end{eqnarray} 
where we use here and below the notation $(X,P)(s)$ in place of $(X,P)(s,t,z)$. Hence, since 
$$  \dot{A}(s,X(s))=\left[\partial_sA+\left(v\cdot\nabla\right)A\right](s,X(s)),
$$
the equation (\ref{Characteristics P}) can be rewritten as
\begin{equation}
\label{Characteristics P Extended}
\dot{P}(s) = \left[-c^{-1}\dot{A}-\nabla\Phi+c^{-1}v\times\left(\nabla\times A\right)+c^{-1}\left(v\cdot\nabla\right)A\right](s,X(s),P(s)).
\end{equation}
The structure of (\ref{Characteristics P Extended}) suggests that we can define a generalized momentum variable $\pi=p+c^{-1}A$ such that the above equation can be reduced to
$$  \dot{\Pi}(s)=\left[-\nabla\Phi+c^{-1}v\times\left(\nabla\times A\right)+c^{-1}\left(v\cdot\nabla\right)A\right](s,X(s),P(s)).
$$ 
Here we have denoted $\Pi(s)=P(s)+c^{-1}A(s,X(s))$. On the other hand, the relativistic velocity written in terms of the generalized momentum is 
\begin{equation}
\label{Relativistic V With A}
  v_A=\frac{\pi-c^{-1}A}{\sqrt{1+c^{-2}\left|\pi-c^{-1}A\right|^2}}.
\end{equation}
Therefore, by using the elementary identity 
$$ v_A\times\left(\nabla\times A\right)+\left(v_A\cdot\nabla\right)A\equiv v^i_A\nabla A^i,
$$
we can reformulate the characteristic system (\ref{Characteristics X})-(\ref{Characteristics P}) in terms of the generalized variables $\xi=(x,\pi)$ as 
\begin{eqnarray}
\label{Characteristics X Gen}
\dot{X}(s,t,\xi) & = & v_A(s,X(s,t,\xi),\Pi(s,t,\xi)),\\
\label{Characteristics P Gen}
\dot{\Pi}(s,t,\xi) & = & -\left[\nabla\Phi-c^{-1}v^i_A\nabla A^i\right](s,X(s,t,\xi),\Pi(s,t,\xi)).
\end{eqnarray} 
As usual, repeated index means summation. Now, standard results in the theory of first order ordinary differential equations imply that for every fixed $t\in I$ and $\xi\in\mathbb{R}^6$ there exists a unique local solution $ \Xi=(X,\Pi)(s,t,\xi)$ of (\ref{Characteristics X Gen})-(\ref{Characteristics P Gen}) satisfying $\Xi(t,t,\xi)=\xi$; see \cite[Chapters II and V]{Hartman}. Moreover, $\Xi\in C^1\left(I\times I\times\mathbb{R}^6;\mathbb{R}^6\right)$. In turn, uniqueness implies that 
$$ Z=(X,\Pi-c^{-1}A)(s,t,x,\pi-c^{-1}A)
$$
is the unique solution of (\ref{Characteristics X})-(\ref{Characteristics P}) with initial data $Z(t,t,,z)=(x,\pi-c^{-1}A)$, so by having the characteristic curves in the generalized phase space we can recover the characteristic curves in the usual phase space.

As the following lemma shows, the field resulting in the right-hand side of the system of equations (\ref{Characteristics X Gen})-(\ref{Characteristics P Gen}) is an incompressible vector field:

\begin{lemma}
\label{Incompressible}
  For $v_A$ given by (\ref{Relativistic V With A}), we have
  $$  \nabla_x\cdot v_A+\nabla_{\pi}\cdot\left(-\nabla\Phi+c^{-1}v^i_A\nabla A^i\right)=0.
  $$
\end{lemma}
\begin{proof}
  Since trivially $\nabla_{\pi}\cdot\nabla\Phi=0$, the result is a consequence of the elementary relation
  $$  c^{-1}\nabla_{\pi}\cdot\left(v^i_A\nabla A^i\right)=c^{-1}\frac{\nabla\cdot A-v^i_A\left(v_A\cdot\nabla\right)A^i}{\sqrt{1+c^{-2}\left|\pi-c^{-1}A\right|}}=-\nabla_x\cdot v_A. 
  $$ 
\end{proof}

As a result, solutions of the characteristic system (\ref{Characteristics X Gen})-(\ref{Characteristics P Gen}) satisfy the volume preserving property. Specifically, for any fixed $s, t\in I$, the map $\Xi(s, t,\cdot):\mathbb{R}^6\rightarrow\mathbb{R}^6$ is a $C^1$-diffeomorphism with inverse $\Xi^{-1}(s,t,\xi)=\Xi(t,s,\xi)$ and Jacobian determinant; see \cite[Corollary V.3.1]{Hartman}
\begin{equation}
\det\frac{\partial\Xi(s,t,\xi)}{\partial\xi}=1.\nonumber
\end{equation}

These properties of the characteristic flow lead to the following result:

\begin{lemma}
\label{Volume Preserving Lemma}
  Let $\left(\Phi,A\right)\in C(I,C^2(\mathbb{R}^3);\mathbb{R}\times\mathbb{R}^3)$ be given in some gauge and let $v_A$ be given by (\ref{Relativistic V With A}). Assume that $\nabla\Phi$ and $\nabla A^i$, $i=1,2,3$ are bounded on $J\times\mathbb{R}^3$ for every compact subinterval $J\subset I$. Let $f_0\in C^1\left(\mathbb{R}^6;\mathbb{R}\right)$ and denote by $\Xi=\left(X,\Pi\right)$ the characteristic flow solving (\ref{Characteristics X Gen})-(\ref{Characteristics P Gen}). Then, the function $f(t,\xi)=f_0(\Xi(0,t,\xi))$ defined on $I\times\mathbb{R}^6$ is the unique $C^1$ solution of the Cauchy problem for
\begin{equation}
\label{Vlasov Potentials}
\partial_tf+v_A\cdot\nabla_xf-\left[\nabla\Phi-c^{-1}v^i_A\nabla A^i\right]\cdot\nabla_{\pi} f=0.
\end{equation}
Moreover, if $f_0\geq0$ then $f\geq0$. Also, for $t\in I$ we have that 
$$  \texttt{supp}f(t)=\Xi(t,\texttt{supp}f_0),
$$  
and for each $1\leq q\leq\infty$, $t\in I$ we have
$$ \left\|f(t)\right\|_{L^q_{x,\pi}}=\left\|f_0\right\|_{L^q_{x,\pi}}.
$$
Conversely, if $f$ is a $C^1$ solution of the Cauchy problem for (\ref{Vlasov Potentials}), then $f$ is constant along each solution of the characteristic system (\ref{Characteristics X Gen})-(\ref{Characteristics P Gen}).
\end{lemma}
\begin{remark} 
\label{Alpha In Vlasov}
In addition, if $\left(\Phi,A\right)(t)\in C^{2,\alpha}(\mathbb{R}^3;\mathbb{R}\times\mathbb{R}^3)$, $0<\alpha<1$, $t\in I$, and $f_0\in C^{1,\alpha}\left(\mathbb{R}^6;\mathbb{R}\right)$, then the unique $C^1$ solution $f(t,\xi)=f_0(\Xi(0,t,\xi))$ of the Cauchy problem for (\ref{Vlasov Potentials}) satisfies $f(t)\in C^{1,\alpha}(\mathbb{R}^6;\mathbb{R})$ for every $t\in I$.  
\end{remark}

\begin{proof}[Proof of Lemma \ref{Volume Preserving Lemma}]
In view of Lemma \ref{Incompressible}, the proof follows by the standard Cauchy's method of characteristics; see \cite[Chapter VI]{Hartman}. In particular, the properties of $f$ are a direct consequence of the properties of the characteristic flow discussed above.   
\end{proof}

We point out that (\ref{Vlasov Potentials}) is the proper Hamiltonian representation of the Vlasov equation (\ref{Vlasov For Maxwell}) in terms of the potentials, since the characteristic equations (\ref{Characteristics X Gen})-(\ref{Characteristics P Gen}) are Hamilton's equations for the Hamiltonian 
\begin{equation}
\label{Hamiltonian}
\mathcal{H}(t,x,\pi)=c^2\sqrt{1+c^{-2}\left|\pi-c^{-1}A(t,x)\right|^2}+\Phi(t,x)
\end{equation}
of a relativistic charged particle under the influence of an electromagnetic field of potentials $(\Phi,A)$. As before, in (\ref{Hamiltonian}) the charge and mass of the particle have been set to one.

%%%%%%%%%%%%%%%%%%%%%%%%%%%%%%%%%%%%%%%%%%%%%%%%%%%%%%%%%% 3 %%%%%%%%%%%%%%%%%%%%%%%%%%%%%%%%%%%%%%%
%%%%%%%%%%%%%%%%%%%%%%%%%%%%%%%%%%%%%%%%%%%%%%%%%%%%%%%%%%%%%%%%%%%%%%%%%%%%%%%%%%%%%%%%%%%%%%%%%%%%

\subsection{The Darwin Potentials}
\label{Darwin Potentials}

To determine the dynamical equations satisfied by the potentials we shall impose the Coulomb gauge condition, since it leads to the Darwin approximation of the Maxwell equations and ultimately to the RVD system. Throughout this section, unless we specify otherwise, we assume that both the charge and current densities $\rho$ and $j$ are smooth and given, and they satisfy the continuity equation 
\begin{equation}
\label{Continuity Equation}
\partial_t\rho+\nabla\cdot j=0.
\end{equation}
Formally, if we substitute the electric and magnetic fields in (\ref{Electric Field})-(\ref{Magnetic Field}) into the non-homogeneous Maxwell equations in (\ref{Maxwell 1})-(\ref{Maxwell 2}), we find that $\Phi$ and $A$ satisfy        
\begin{eqnarray}
  \Delta\Phi & = & -4\pi \rho-c^{-1}\partial_t\left(\nabla\cdot A\right),\\
  \Delta A-c^{-2}\partial_t^2A & = & -c^{-1}4\pi j+\nabla\left(\nabla\cdot A+c^{-1}\partial_t\Phi\right).
\end{eqnarray}
Therefore, in the Coulomb gauge (\ref{Coulomb Gauge}), the potentials satisfy 
\begin{eqnarray}
\label{Poisson Scalar Maxwell}
  \Delta\Phi & = & -4\pi\rho,\\
\label{Poisson Vector Maxwell}
  \Delta A-c^{-2}\partial_t^2A & = & -c^{-1}4\pi j+c^{-1}\nabla\partial_t\Phi.
\end{eqnarray}
On the other hand, any smooth solution $(\Phi,A)$ of the above system that satisfies the Coulomb gauge condition initially, will continue to do so for all times, and therefore the induced electromagnetic field will solve (\ref{Maxwell 1})-(\ref{Maxwell 2}). Indeed, if $(\Phi,A)$ is a smooth solution of (\ref{Poisson Scalar Maxwell})-(\ref{Poisson Vector Maxwell}) that satisfies $\left.\nabla\cdot A\right|_{t=0}=0$ and $\left.\partial_t(\nabla\cdot A)\right|_{t=0}=0$, then $g_C=\nabla\cdot A$ is the solution of
\begin{eqnarray} 
\Delta g_C-c^{-2}\partial^2_tg_C =-4\pi c^{-1}\left(\nabla\cdot j+\partial_t\rho\right)=0,\nonumber\\
\left.g_C\right|_{t=0}=0,\quad \left.\partial_tg_C\right|_{t=0}=0,\hspace{2.6cm}\nonumber
\end{eqnarray}
and the claim follows. Hence, the system of equations (\ref{Poisson Scalar Maxwell})-(\ref{Poisson Vector Maxwell}) complemented with (\ref{Coulomb Gauge}) is fully equivalent to the set of Maxwell equations (\ref{Maxwell 1})-(\ref{Maxwell 2}). 

We define the Darwin approximation of the Maxwell equations as the quasi-static limit of the system (\ref{Poisson Scalar Maxwell})-(\ref{Poisson Vector Maxwell}):
\begin{definition}
\label{Darwin}
    Let $(\rho,j):I\times\mathbb{R}^3\rightarrow\mathbb{R}\times\mathbb{R}^3$ be given and satisfy the continuity equation (\ref{Continuity Equation}). The set of potentials $(\Phi,A)$ is called a classical solution of the Darwin equations if $\Phi\in C^1(I,C^2(\mathbb{R}^3);\mathbb{R})$, $A\in C(I,C^2(\mathbb{R}^3);\mathbb{R}^3)$ and, on $I\times\mathbb{R}^3$, 
\begin{eqnarray}
\label{Poisson Scalar Darwin}
  \Delta\Phi & = & -4\pi \rho,\\
\label{Poisson Vector Darwin}
  \Delta A & = & -c^{-1}4\pi j+c^{-1}\nabla\partial_t\Phi.
\end{eqnarray}
\end{definition}  

The system (\ref{Poisson Scalar Darwin})-(\ref{Poisson Vector Darwin}) has the following explicit solution, as proved below:

\begin{definition} 
\label{Darwin Potentials Definition}
For the charge and current densities $(\rho,j):I\times\mathbb{R}^3\rightarrow\mathbb{R}\times\mathbb{R}^3$ we formally define the set of Darwin potentials $(\Phi_D,A_D):I\times\mathbb{R}^3\rightarrow\mathbb{R}^3\times\mathbb{R}^3$ by
\begin{eqnarray}
\label{Scalar Potential Darwin}
  \Phi_D(t,x) & = & \int_{\mathbb{R}^3}\rho(t,y)\frac{dy}{\left|y-x\right|},\\
\label{Vector Potential Darwin}
  A_D(t,x) & = & \frac{1}{2c}\int_{\mathbb{R}^3}\left[\texttt{id}+\omega\otimes\omega\right]j(t,y)\frac{dy}{\left|y-x\right|},
\end{eqnarray}
where $\omega=\left(y-x\right)/\left|y-x\right|$ and $\texttt{id}$ denotes the identity matrix.
\end{definition}

\begin{lemma}
\label{Darwin Potentials Lemma}
 Let $\rho\in C^1(I,C^{\alpha}_0(\mathbb{R}^3);\mathbb{R})$ and $j\in C(I,C^{1,\alpha}_0(\mathbb{R}^3);\mathbb{R}^3)$, $0<\alpha<1$, be given -they do not need to satisfy the continuity equation (\ref{Continuity Equation})-. Define the field
\begin{equation}
\label{Transversal Current}
   \mathbb{P}j(t,x)=j(t,x)+\frac{1}{4\pi}\nabla\int_{\mathbb{R}^3}\nabla\cdot j(t,y)\frac{dy}{\left|y-x\right|},\quad t\in I,\quad x\in\mathbb{R}^3.
\end{equation}
Then the following holds: 
    \begin{description}
       \item [(a)] The scalar potential $\Phi_D$ is the unique solution in $C^1(I,C^{2,\alpha}(\mathbb{R}^3);\mathbb{R})$ of 
          \begin{equation}
          \label{Poisson Phi}
           \Delta\Phi(t,x)=-4\pi\rho(t,x),\quad \lim_{\left|x\right|\rightarrow\infty}\Phi(t,x)=0.
           \end{equation}
          It satisfies 
           $$ \nabla\Phi_D(t,x)=\int_{\mathbb{R}^3}\rho(t,y)\frac{\omega dy}{\left|y-x\right|^2}. 
           $$
       \item [(b)] $\mathbb{P}j\in C(I,C^{1,\alpha}(\mathbb{R}^3);\mathbb{R}^3)$. It satisfies $\nabla\cdot\mathbb{P}j=0$ (i.e., $\mathbb{P}j$ is the transversal component of the current density $j$), and $\mathbb{P}j(x)=O(\left|x\right|^{-2})$ for $\left|x\right|\rightarrow\infty$. 
       \item [(c)] The vector potential $A_D$ is the unique solution in $C(I,C^{3,\alpha}(\mathbb{R}^3);\mathbb{R}^3)$ of 
          \begin{equation}
          \label{Poisson A}
            \Delta A(t,x)=-4\pi c^{-1}\mathbb{P}j(t,x),\quad \lim_{\left|x\right|\rightarrow\infty}\left|A(t,x)\right|=0.
           \end{equation}
          It satisfies 
           \begin{equation}
           \label{Derivative A Representation}
           \partial_xA_D(t,x)=\frac{1}{2c}\int_{\mathbb{R}^3}\left\{\omega\otimes j-j\otimes\omega+\left[3\omega\otimes\omega-\texttt{id}\right]\left(j\cdot\omega\right)\right\}\frac{dy}{\left|y-x\right|^2}, 
           \end{equation}     
with $j=j(t,y)$. In particular,
          $$\nabla\cdot A_D(t,x)=0 \quad \hbox{and} \quad \nabla\times A_D(t,x)=\frac{1}{c}\int_{\mathbb{R}^3}\omega\times j(t,y)\frac{dy}{\left|y-x\right|^2}.$$ 
    \end{description} 
\end{lemma}

\begin{corollary} 
\label{Corollary Darwin Potentials Solutions}
If $\rho$ and $j$, as given in Lemma \ref{Darwin Potentials Lemma}, satisfy the continuity equation (\ref{Continuity Equation}), then 
       $$ \mathbb{P}j(t,x)=j(t,x)-\frac{1}{4\pi}\nabla\partial_t\Phi_D(t,x),\quad t\in I,\quad x\in\mathbb{R}^3,$$
and thus the Darwin potentials (\ref{Scalar Potential Darwin})-(\ref{Vector Potential Darwin}) are the unique classical solution of the Darwin equations (\ref{Poisson Scalar Darwin})-(\ref{Poisson Vector Darwin}). 
\end{corollary}

\begin{proof}[Proof of Lemma \ref{Darwin Potentials Lemma}]
Without loss of generality we omit the time dependence. 

The proof of (a) is a standard result for the Poisson equation. Existence (in a much weaker sense) can be found, for instance, in \cite[Theorem 6.21]{Lieb} while the regularity of the solution is given in \cite[Theorem 10.3]{Lieb}. Uniqueness is known as Liouville's theorem \cite[Theorem 7 Section 4.2]{McOwen}.

To prove (b), notice that $\nabla\cdot j\in C^{\alpha}_0(\mathbb{R}^3;\mathbb{R})$. Hence, as in (a), the integral in the right-hand side of (\ref{Transversal Current}) is the $C^{2,\alpha}$-solution of the Poisson equation $\Delta u=-4\pi\nabla\cdot j$, $\lim_{\left|x\right|\rightarrow\infty}u(x)=0$. That $\nabla u(x)=O(\left|x\right|^{-2})$ for $\left|x\right|\rightarrow\infty$ is well known, which in turn provides the decay for $\mathbb{P}j$, since $j$ has compact support. Moreover, 
$$ \nabla\cdot\mathbb{P}j(x) = \nabla\cdot j(x) +\frac{1}{4\pi}\Delta\int_{\mathbb{R}^3}\nabla\cdot j(y)\frac{dy}{\left|y-x\right|} = 0.
$$

As for (c), we first prove the following lemma:

\begin{lemma}
\label{Equivalent Representation A Lemma}
   The Darwin potential $A_D$ in (\ref{Vector Potential Darwin}) has the equivalent representation
\begin{equation}
\label{Equivalent Representation A} A_D(x)=\frac{1}{c}\int_{\mathbb{R}^3}j(y)\frac{dy}{\left|y-x\right|}+\frac{1}{2c}\int_{\mathbb{R}^3}\nabla\cdot j(y)\frac{y-x}{\left|y-x\right|}dy.
\end{equation}
\end{lemma} 
\begin{proof} The current density $j$ has compact support, so standard arguments can show that the right-hand side (\texttt{RHS}) of the above expression is well defined. The divergence theorem then yields, 
\begin{eqnarray}
  \texttt{RHS} & = & \frac{1}{c}\int_{\mathbb{R}^3}j(y)\frac{dy}{\left|y-x\right|}-\frac{1}{2c}\int_{\mathbb{R}^3}\left( j(y)\cdot\nabla\right)\omega dy\nonumber\\
  & = & \frac{1}{c}\int_{\mathbb{R}^3}\left\{j(y)-\frac{1}{2}\left[j(y)-\omega\left(j(y)\cdot\omega\right)\right]\right\}\frac{dy}{\left|y-x\right|}\nonumber\\
  & = & \frac{1}{2c}\int_{\mathbb{R}^3}\left[j(y)+\omega\left(j(y)\cdot\omega\right)\right]\frac{dy}{\left|y-x\right|},\nonumber
\end{eqnarray}
which is precisely the Darwin potential $A_D$ in (\ref{Vector Potential Darwin}). The use of the divergence theorem is justified by the following standard argument: remove a small ball about $x\in\mathbb{R}^3$ in the domain of integration so we can avoid the singularity at $y=x$, then use the divergence theorem and note that the boundary term corresponding to the small ball vanishes as its radius tends to $0$.
\end{proof}

We shall now deduce by direct computation from (\ref{Equivalent Representation A}), the Poisson equation given by (\ref{Poisson A}). To this end, we first recall that just as in part (a), 
\begin{equation}
\label{Delta 1} \Delta\left\{\frac{1}{c}\int_{\mathbb{R}^3}j(y)\frac{dy}{\left|y-x\right|}\right\}=-\frac{4\pi}{c}j(x).
\end{equation}
The integral in curly brackets is in $C^{3,\alpha}(\mathbb{R}^3;\mathbb{R}^3)$, due to the regularity of $j$. Next, we show that the following equality holds in the sense of distributions,
\begin{equation}
\label{Delta Intermediate 1}
\partial_{x_k}\left\{\int_{\mathbb{R}^3}\nabla\cdot j(y)\omega^idy\right\} = -\int_{\mathbb{R}^3}\nabla\cdot j(y)\left[\delta_{ik}-\omega^i\omega^k\right]\frac{dy}{\left|y-x\right|}.
\end{equation}

Let $r=\left|y-x\right|>0$. First, note that $\partial_{x_k}\omega^i=-r^{-1}\left[\delta_{ik}-\omega^i\omega^k\right]$, and that the integral on the right-hand side of (\ref{Delta Intermediate 1}) is well defined for almost all $x\in\mathbb{R}^3$, since $\nabla\cdot j\in C_0^{\alpha}(\mathbb{R}^3,\mathbb{R})$ and the kernel is bounded from above by $r^{-1}$. For $\phi\in C^{\infty}_0(\mathbb{R}^3;\mathbb{R})$, the function  $(x,y)\mapsto\partial_{x_k}\phi(x)\nabla\cdot j(y)\omega^i$ is integrable on $\mathbb{R}^3\times\mathbb{R}^3$. Hence, we can use Fubini's theorem to find that
\begin{eqnarray}
\int_{\mathbb{R}^3}\partial_{x_k}\phi(x)\left\{\int_{\mathbb{R}^3}\nabla\cdot j(y)\omega^idy\right\}dx & = & \int_{\mathbb{R}^3}\left\{\int_{\mathbb{R}^3}\left(\partial_{x_k}\phi(x)\right)\omega^idx\right\}\nabla\cdot j(y)dy\nonumber\\
 & = &-\int_{\mathbb{R}^3}\left\{\int_{\mathbb{R}^3}\phi(x)\partial_{x_k}\omega^idx\right\}\nabla\cdot j(y)dy,\nonumber
\end{eqnarray}
where the second equality is justified by a standard limiting process and integrations by parts, similar to the argument at the end of the proof of Lemma \ref{Equivalent Representation A Lemma}. Then, another use of Fubini's theorem yields (\ref{Delta Intermediate 1}) in the sense of distributions, as claimed. 

Actually, the equality in (\ref{Delta Intermediate 1}) holds in the classical sense. By the standard theory of the Poisson equation, the right-hand side of (\ref{Delta Intermediate 1}) is a function in $C^{2,\alpha}(\mathbb{R}^3;\mathbb{R})$; see \cite[Theorem 10.3]{Lieb}. Therefore, in view of the theorem for the equivalence of classical and distributional derivatives, the integral in curly brackets on the left-hand side of (\ref{Delta Intermediate 1}) is in $C^{3,\alpha}(\mathbb{R}^3;\mathbb{R})$; see \cite[Theorem 6.10]{Lieb}.  

Now, since $\partial_{x_k}r^{-1}=r^{-2}\omega^k$ and $\omega^k\left[\delta_{ik}-\omega^i\omega^k\right]\equiv0$, we have
$$ \partial_{x_k}\left\{r^{-1}\left[\delta_{ik}-\omega^i\omega^k\right]\right\}=-r^{-1}\omega^i\partial_{x_k}\omega^k=2r^{-2}\omega^i.
$$
Therefore, similar arguments to those used above yield
\begin{eqnarray}
\label{Delta Intermediate 2}
\lefteqn{\partial_{x_k}\left\{-\int_{\mathbb{R}^3}\nabla\cdot j(y)\left[\delta_{ik}-\omega^i\omega^k\right]\frac{dy}{\left|y-x\right|}\right\}}\nonumber\\
& = & -2\int_{\mathbb{R}^3}\nabla\cdot j(y)\frac{\omega^idy}{\left|y-x\right|^2}= -2\partial_{x_i}\int_{\mathbb{R}^3}\nabla\cdot j(y)\frac{dy}{\left|y-x\right|}.\hspace{.3cm}
\end{eqnarray}
Hence, since $\Delta\equiv\nabla\cdot\nabla$, we can combine (\ref{Delta Intermediate 1}) and (\ref{Delta Intermediate 2}) to find that
\begin{eqnarray} 
\label{Delta 2}
\Delta\left\{\frac{1}{2c}\int_{\mathbb{R}^3}\nabla\cdot j(y)\frac{y-x}{\left|y-x\right|}dy\right\} 
& = & -\frac{1}{c}\nabla\int_{\mathbb{R}^3}\nabla\cdot j(y)\frac{dy}{\left|y-x\right|}.
\end{eqnarray}
Then, we add (\ref{Delta 1}) and (\ref{Delta 2}) to conclude that $\Delta A_D=-4\pi c^{-1}\mathbb{P}j$ holds on $\mathbb{R}^3$, and so $A_D$ is a $C^{3,\alpha}$ solution of (\ref{Poisson A}). This solution is unique in view of the Liouville's theorem \cite[Theorem 7 Section 4.2]{McOwen}.

The representation (\ref{Derivative A Representation}) of $\partial_xA$ can be proved as follows. Since $j_A$ is regular enough, we shift the $x$-variable into the argument of $j_A$ and differentiate (\ref{Vector Potential Darwin}) under the integral. Then, we move the derivative to the kernel of (\ref{Vector Potential Darwin}) helped by the same standard argument at the end of the proof of Lemma \ref{Equivalent Representation A Lemma}. In doing so, we notice that for $r>0$, the $imk$-th entry of $\partial_x\mathcal{K}$ is
$$\partial_{x_k}\left\{r^{-1}\left[\delta_{im}+\omega^i\omega^m\right]\right\}=r^{-2}\left[\delta_{im}\omega^k-\delta_{km}\omega^i-\delta_{ik}\omega^m+3\omega^i\omega^k\omega^m\right],$$ 
where $\mathcal{K}=r^{-1}\left[\texttt{id}+\omega\otimes\omega\right]$. This leads to (\ref{Derivative A Representation}). Finally, it is not difficult to check that
$$\nabla\cdot A_D=\texttt{Trace}(\partial_xA_D)=0\quad \hbox{and}\quad \left(\nabla\times A_D\right)^i=\frac{1}{2}\left(\partial_xA_D-(\partial_xA_D)^T\right)_{kl},$$ where $(\partial_xA_D)^T$ denotes the transpose of $\partial_xA_D$, and $i,k,l\in\left\{1,2,3\right\}$ are given according to the cyclic index-permutation.
\end{proof}

It is easy to check that the Darwin equations in Definition \ref{Darwin} are formally equivalent to the equations (\ref{Components Electric Field}) and (\ref{Darwin 1})-(\ref{Darwin 2}) given in the Introduction. To see this let us define $$ E_L=-\nabla\Phi_D, \quad E_T=-c^{-1}\partial_tA_D, \quad B=\nabla\times A_D.
$$
We have to show that $(E_L,E_T,B)$ formally solves (\ref{Darwin 1})-(\ref{Darwin 2}) provided that the charge and current densities satisfy the continuity equation. Clearly $\nabla\cdot B=0$, and since $\nabla\cdot A_D=0$, we have 
\begin{eqnarray}
\nabla\times B = \nabla\left(\nabla\cdot A_D\right)-\Delta A_D = 4\pi c^{-1}j-c^{-1}\partial_t\nabla\Phi_D =4\pi c^{-1}j+c^{-1}\partial_tE_L,\nonumber
\end{eqnarray}
which is (\ref{Darwin 1}). Easy computations yield (\ref{Components Electric Field}) and (\ref{Darwin 2}), and the claim follows. 

We conclude this section with a-priori estimates on the potentials and their \textit{space} derivatives. For simplicity and without loss of generality, we shall neglect the time dependence.

\begin{lemma} 
\label{Pallard}
For $1\leq m<3$ set $r_0=3/(3-m)$ and let $r<r_0<s$. Then there exists a positive constant $C=C(m,r,s)$ such that for any $\Psi\in L^r\cap L^s(\mathbb{R}^3;\mathbb{R})$ 
$$ \left\|\int_{\mathbb{R}^3}\Psi(y)\frac{dy}{\left|y-\cdot\right|^m}\right\|_{L^{\infty}_x}\leq C(m,r,s)\left\|\Psi\right\|^{1-\lambda}_{L^{r}_x}\left\|\Psi\right\|^{\lambda}_{L^{s}_x},\quad \hbox{where}\quad\lambda=\frac{1-r/r_0}{1-r/s}.
$$
In particular, $C(m,1,\infty)=3\left(4\pi/m\right)^{m/3}/\left(3-m\right)$.
\end{lemma}
\begin{proof}
  See \cite[Lemma 2.7]{Pallard}.
\end{proof}

\begin{lemma} 
\label{Estimate Darwin Potentials}
For $\rho$ and $j$ as given in Lemma \ref{Darwin Potentials Lemma}, the Darwin vector potential (\ref{Vector Potential Darwin}) satisfy the estimates:
\begin{equation}
\label{Estimates Vector Potentials First Derivatives}
\left\|A_D\right\|_{L^{\infty}_x}\leq C\left\|j\right\|^{2/3}_{L^1_x}\left\|j\right\|^{1/3}_{L^{\infty}_x}\quad  \hbox{and}\quad \left\|\partial_xA_D\right\|_{L^{\infty}_x}\leq C\left\|j\right\|^{1/3}_{L^1_x}\left\|j\right\|^{2/3}_{L^{\infty}_x}.
\end{equation}
Moreover, for any $0<h\leq R$ we have
$$
  \left\|\partial^2_xA_D\right\|_{L^{\infty}_x} \leq  C\left[R^{-3}\left\|j\right\|_{L^1_x}+h\left\|\partial_xj\right\|_{L^{\infty}_x}+\left(1+\ln\left(R/h\right)\right)\left\|j\right\|_{L^{\infty}_x}\right],
$$
where $C>0$ is independent of $h$, $R$, $\rho$ and $j$. In particular,
\begin{equation}
\label{Estimates Vector Potentials Second Derivatives}
 \left\|\partial^2_xA_D\right\|_{L^{\infty}_x} \leq  C\left[\left\|j\right\|_{L^1_x}+\left(1+\left\|j\right\|_{L^{\infty}_x}\right)\left(1+\ln^{+}\left\|\partial_xj\right\|_{L^{\infty}_x}\right)\right].
\end{equation}
The same estimates hold for the scalar potential $\Phi_D$, with $j$ replaced by $\rho$.
\end{lemma}

\begin{proof} The estimates corresponding to $\Phi_D$ are well known from the study of the Vlasov-Poisson system. These results can be found, for instance, in \cite[Lemma P1]{Rein} and \cite[Propositions 1 and 2]{Batt}. Here, we shall produce the estimates for the vector potential $A_D$ only. The proof is actually rather similar.

Let $\mathcal{K}(y,x)=\left|y-x\right|^{-1}\left[\texttt{id}+\omega\otimes\omega\right]$. Clearly, $\left|\mathcal{K}(y,x)\right|\leq C\left|y-x\right|^{-1}$. Then, the estimates in (\ref{Estimates Vector Potentials First Derivatives}) are a straightforward consequence of Lemma \ref{Pallard}. To produce the estimates for the second derivatives, consider 
\begin{eqnarray}
   \partial_l\partial_kA^i_D & \equiv & \frac{1}{2c} \left\{\partial_l\int_{\mathbb{R}^3}\left[\delta_{im}\omega^k-\delta_{km}\omega^i-\delta_{ik}\omega^m\right]j^m(y)\frac{dy}{\left|y-x\right|^2}\right.\nonumber\\
   & & \left. 3\partial_l\int_{\mathbb{R}^3}j^m\frac{\omega^m\omega^i\omega^kdy}{\left|y-x\right|^2}\right\}\nonumber\\
   & = & \frac{1}{2c}\left(I_1+3I_2\right).\nonumber
\end{eqnarray}
Here we have introduced the notation $\partial_k=\partial_{x_k}$, $k=1,2,3$; see Lemma \ref{Darwin Potentials Lemma}(c) for the matrix representation of the integrand of $\partial_xA$. Now, the integral $I_1$ can in turn be split into three integrals, each one essentially the same as the integral corresponding to $\partial_l\partial_k\Phi_D$. Thus, $I_1$ satisfies the expected estimates, as $\partial_l\partial_k\Phi_D$ does. Therefore, we are led to estimate $I_2$. To this end, we set $r=\left|y-x\right|$, and for $r>0$ we denote by $\Gamma^m_{ikl}(y-x)$ 
$$ \partial_{y_l}\left[\frac{\omega^m\omega^i\omega^k}{r^2}\right]= \frac{1}{r^3}\left[\delta_{ml}\omega^i\omega^k+\delta_{il}\omega^k\omega^m+\delta_{kl}\omega^i\omega^m-5\omega^i\omega^k\omega^l\omega^m\right].
$$
This kernel is too singular to use Lemma \ref{Pallard}. However, since $y^iy^ky^m\left|y\right|^{-5}$ is homogeneous of degree $-2$, for every $0< R_1<R_2$ we have 
$$ \int_{R_1<\left|y\right|<R_2}\Gamma^m_{ikl}(y)dy = \int_{\left|y\right|=R_2}\frac{y^l}{R_2}\frac{y^iy^ky^m}{\left|y\right|^5}dS_y-\int_{\left|y\right|=R_1}\frac{y^l}{R_1}\frac{y^iy^ky^m}{\left|y\right|^5}dS_y=0. 
$$
Thus, for any $h>0$, we can rewrite $I_2$ as
\begin{eqnarray}
 I_2 & = & \int_{\left|y-x\right|>h}\Gamma^m_{ikl}(y-x)j^m(y)dy+j^m(x)\int_{\left|\omega\right|=1}\omega^i\omega^k\omega^l\omega^md\omega\nonumber\\
& & +\int_{\left|y-x\right|\leq h}\Gamma^m_{ikl}(y-x)\left[j^m(y)-j^m(x)\right]dy.\nonumber
\end{eqnarray}
The singularity in the last integral at $r=0$ is now avoided by the difference $j^m(y)-j^m(x)$. Indeed, for $0< h\leq R$ we produce
\begin{eqnarray}
   I_2 & \leq & C\left\{\left\|j\right\|_{L^{\infty}_x}\int_{h<\left|y-x\right|\leq R}\frac{dy}{\left|y-x\right|^3}+\int_{\left|y-x\right|>R}\left|j(y)\right|\frac{dy}{\left|y-x\right|^3}\right.\nonumber\\
   & & \hspace{.6cm}\left.\left\|\partial_xj\right\|_{L^{\infty}_x}\int_{\left|y-x\right|\leq h}\frac{dy}{\left|y-x\right|^2}+\left|j(x)\right|\right\}\nonumber\\
   &\leq & C\left[\ln(R/h)\left\|j\right\|_{L^{\infty}_x}+R^{-3}\left\|j\right\|_{L^1_x}+h\left\|\partial_xj\right\|_{L^{\infty}_x}+\left\|j\right\|_{L^{\infty}_x}\right].\nonumber
\end{eqnarray}
This yields the first estimate on $\left\|\partial^2_xA_D\right\|_{L^{\infty}_x}$. Then, by setting $R=1$ and letting $h=\left\|\partial_xj\right\|^{-1}_{L^{\infty}_x}$ if $\left\|\partial_xj\right\|^{-1}_{L^{\infty}_x}\geq1$, otherwise $h=1$, the estimate (\ref{Estimates Vector Potentials Second Derivatives}) follows as well. This completes the proof of the lemma.
\end{proof}

%%%%%%%%%%%%%%%%%%%%%%%%%%%%%%%%%%%%%%%%%%%%%%%%%%%%%%%%%% 4 %%%%%%%%%%%%%%%%%%%%%%%%%%%%%%%%%%%%%%%%%%%%
%%%%%%%%%%%%%%%%%%%%%%%%%%%%%%%%%%%%%%%%%%%%%%%%%%%%%%%%%%%%%%%%%%%%%%%%%%%%%%%%%%%%%%%%%%%%%%%%%%%%%%%%%
\section{The RVD System}
\label{The RVD System Section}

If we now combine (\ref{Vlasov Potentials}) and (\ref{Scalar Potential Darwin})-(\ref{Vector Potential Darwin}) by means of (\ref{Density and Current}), then we obtain the following equivalent representation of the RVD system: 
\begin{eqnarray}
\label{Vlasov In RVD}
\partial_tf+v_A\cdot\nabla_xf-\left[\nabla\Phi-c^{-1}v^i_A\nabla A^i\right]\cdot\nabla_{p}f=0,
\end{eqnarray}
coupled with
\begin{eqnarray}
\label{Phi In RVD}
\Phi(t,x) & = & \int_{\mathbb{R}^3}\int_{\mathbb{R}^3}f(t,y,p)\frac{dpdy}{\left|y-x\right|},\\
\label{A In RVD}
A(t,x) & = & \frac{1}{2c}\int_{\mathbb{R}^3}\int_{\mathbb{R}^3}\left[\texttt{id}+\omega\otimes\omega\right]v_Af(t,y,p)\frac{dpdy}{\left|y-x\right|},
\end{eqnarray}
where
\begin{equation}
\label{Velocity In RVD}
v_A=\frac{p-c^{-1}A}{\sqrt{1+c^{-2}\left|p-c^{-1}A\right|^2}}.
\end{equation}

For the sake of notation, we have written $p$ instead of $\pi$ when referring to the generalized momentum variable. We will continue to do so for the rest of the paper. We shall also set $c=1$ for the speed of light. The goal is to prove that a small enough Cauchy datum launches a unique classical solution of the system (\ref{Vlasov In RVD})-(\ref{Velocity In RVD}) globally in time. We shall prove this in Subsections \ref{Local Solutions Section} and \ref{Global Solutions Section} below, but first we center our attention on (\ref{A In RVD}). If $f$ is given, then (\ref{A In RVD}) is a nonlinear integral equation of unknown $A$.

\begin{lemma} 
\label{Fixed Point Lemma}
Fix $t\in I$ and let $f(t)\in C^{1,\alpha}_0(\mathbb{R}^6;\mathbb{R})$, $0<\alpha<1$, be given. Then there exists a unique $A(t)\in C_b\cap C^{2,\alpha}(\mathbb{R}^3;\mathbb{R}^3)$ satisfying (\ref{A In RVD})-(\ref{Velocity In RVD}).
\end{lemma}   

\begin{proof} Without loss of generality we shall omit the dependence in time. Let $\bar{C}$ be a constant that may depend on $f$, to be fixed later on. Define the set
$$ \mathcal{D}_{\bar{C}}=\left\{A\in C_b(\mathbb{R}^3;\mathbb{R}^3):\left\|A\right\|_{L^{\infty}_x}\leq \bar{C}\right\}. 
$$
First, we show that there exists an $A_{\infty}\in \mathcal{D}_{\bar{C}}$ which solves (\ref{A In RVD})-(\ref{Velocity In RVD}). To this end, denote the kernel $\mathcal{K}(x,y)=\left|y-x\right|^{-1}\left[\texttt{id}+\omega\otimes\omega\right]$ and let $A\in \mathcal{D}_{\bar{C}}$. Consider the mapping $A\mapsto T[A]$ defined by 
$$
T[A](x)= \frac{1}{2}\int_{\mathbb{R}^3}\int_{\mathbb{R}^3}\mathcal{K}(x,y)v_Af(y,p)dpdy, \quad v_A=\frac{p-A}{\sqrt{1+\left|p-A\right|^2}}. 
$$
We claim that $T[A]\in\mathcal{D}_{\bar{C}}$. Indeed, let $(\mathcal{K})_{ij}(x,y)$ be the $ij$-entry of $\mathcal{K}(x,y)$. For some $u_1$, $u_2$ and $u_3$ on the line segment between $x$ and $z$, the mean value theorem implies
\begin{eqnarray}
\left|\left(\mathcal{K}\right)_{ij}(x,y)-\left(\mathcal{K}\right)_{ij}(z,y)\right| & \leq & \left|\frac{1}{\left|y-x\right|}-\frac{1}{\left|y-z\right|}\right|+\left|\frac{y^i-x^i}{\left|y-x\right|^2}-\frac{y^i-z^i}{\left|y-z\right|^2}\right|\nonumber\\
& & +\left|\frac{y^j-x^j}{\left|y-x\right|^2}-\frac{y^j-z^j}{\left|y-z\right|^2}\right|\nonumber\\
& \leq & C\left|x-z\right|\left(\frac{1}{\left|y-u_1\right|^2}+\frac{1}{\left|y-u_2\right|^2}+\frac{1}{\left|y-u_3\right|^2}\right).\nonumber
\end{eqnarray}
Hence, since $\left|v_A\right|\leq 1$, a use of Lemma \ref{Estimate Darwin Potentials} produces
\begin{eqnarray}
\label{Continuous Map}
 \left|T[A](x)-T[A](z)\right| & \leq & \frac{1}{2}\int_{\mathbb{R}^3}\left|\mathcal{K}(x,y)-\mathcal{K}(z,y)\right|\rho(y)dy\nonumber\\
 & \leq & C\left|x-z\right|\left\|\int_{\mathbb{R}^3}\rho(y)\frac{dy}{\left|y-\cdot\right|^2}\right\|_{L^{\infty}_x}\nonumber\\
 & \leq & C(f)\left|x-z\right|.
\end{eqnarray}  
Thus, $T[A]$ is a continuous vector valued function. Also, by Lemma \ref{Estimate Darwin Potentials} 
\begin{equation}
\label{Bounded Map}
 \left\|T[A]\right\|_{L^{\infty}_x}\leq 3/2(\pi/2)^{1/3}\left\|\rho\right\|^{2/3}_{L^1_x}\left\|\rho\right\|^{1/3}_{L^{\infty}_x}\equiv \bar{C}.
\end{equation}
Therefore, $T[A]\in\mathcal{D}_{\bar{C}}$ as claimed.

We now show that $T$ has a fixed point $A_{\infty}\in \mathcal{D}_{\bar{C}}$. By virtue of the Schauder fixed point theorem \cite[Theorem 3 Section 9.1]{{McOwen}}, it suffices to show that $T$ is a continuous mapping and that the closure of the image of $T$ is compact in $\mathcal{D}_{\bar{C}}$. To show the continuity of $T$, we see that if $A_k\rightarrow A$ in $\mathcal{D}_{\bar{C}}$, then by Lemma \ref{Estimate Darwin Potentials} 
\begin{eqnarray}
\left|T[A_k](x)-T[A](x)\right| & \leq & C\int_{\mathbb{R}^3}\int_{\mathbb{R}^3}\left|v_{A_k}-v_A\right|f(y,p)\frac{dpdy}{\left|y-x\right|}\nonumber\\
& \leq & C(f)\left\|A_k-A\right\|_{L^{\infty}_x}.\nonumber \end{eqnarray}
To show that $\overline{T\mathcal{D}_{\bar{C}}}\subset \mathcal{D}_{\bar{C}}$ is compact, we first notice that for $A\in \mathcal{D}_{\bar{C}}$,
\begin{equation} 
\label{Decay}
\left|T\left[A\right](x)\right|\leq \left\|\rho\right\|_{L^{\infty}_x}\int_{\texttt{supp}f}\frac{dy}{\left|x-y\right|}\leq C(f)\frac{1}{1+\left|x\right|}.
\end{equation}
Now consider the sequence $\left\{B_{n}\right\}\subset T\mathcal{D}_{\bar{C}}$. Let $R>0$ be fixed. By (\ref{Bounded Map}) and (\ref{Continuous Map}), the restriction $$\left.\left\{B_{n}\right\}\right|_{\left\{x\in\mathbb{R}^3:\left|x\right|\leq R\right\}}$$ is clearly bounded and equicontinuous. Then, by Arzel{\`a}-Ascoli and a standard diagonal argument, we can find a subsequence $\left\{B_{n_k}\right\}$ and a continuous, bounded limit vector field $B$ such that $\left\{B_{n_k}\right\}\rightarrow B$ uniformly on compact sets, and in particular pointwise. Clearly, $\left\|B\right\|_{L^{\infty}_x}\leq \bar{C}$, and since $\left\{B_{n_k}\right\}$ satisfies the estimate (\ref{Decay}), so does $B$. We only need to show that the convergence $\left\{B_{n_k}\right\}\rightarrow B$ is uniform. To this end, let $\epsilon>0$. Choose $R>0$ such that the right-hand side of (\ref{Decay}) is less than $\epsilon/2$ for $\left|x\right|>R$. Then, for all $k$ we have $\left|B_{n_k}(x)-B(x)\right|<\epsilon$ for $\left|x\right|>R$, and we can find a $k_0=k_0(R,\epsilon)$ such that for all $k>k_0$
$$ \sup_{\left|x\right|\leq R}\left|B_{n_k}(x)-B(x)\right|<\epsilon.
$$
This proves uniform convergence. Hence, all the hypotheses for the Schauder fixed point theorem are fulfilled, and thus $T$ has a fixed point $A_{\infty}$ in $\mathcal{D}_{\bar{C}}$.  

Next, we show that $A_{\infty}$ has the required regularity. To this end, define $v_{A_{\infty}}$ and then $j_{A_{\infty}}$ according to (\ref{Velocity In RVD}) and (\ref{Density and Current}), respectively. The vector field $A_{\infty}$ has the form of a Darwin potential (\ref{Vector Potential Darwin}) with current density $j_{A_{\infty}}\in C_0(\mathbb{R}^3;\mathbb{R}^3)$. Clearly, the kernel of (\ref{Vector Potential Darwin}) satisfies $\left|\mathcal{K}(x,y)\right|\leq C\left|y-x\right|^{-1}$ and the derivative estimate $\left|\partial_x\mathcal{K}(x,y)\right|\leq C\left|y-x\right|^{-2}$. Hence, we can use the standard theory for the Poisson equation to find that $A_{\infty}\in C^1(\mathbb{R}^3;\mathbb{R}^3)$; see, for instance, \cite[Lemma 4.1]{Trudinger} or \cite[Theorem 10.2 (iii)]{Lieb}. But such a regularity of $A_{\infty}$ implies that $j_{A_{\infty}}\in C^1_0(\mathbb{R}^3;\mathbb{R}^3)$. Thus, we also have $j_{A_{\infty}}\in C^{\alpha}_0(\mathbb{R}^3;\mathbb{R}^3)$, $0<\alpha<1$ and so $A_{\infty}\in C^{2,\alpha}(\mathbb{R}^3;\mathbb{R}^3)$, as desired. For the latter implication see, for instance, \cite[Theorem 10.3]{Lieb}.

%%%%%%%%%%%%%%%%%%%%%%%%%%

Now we prove the uniqueness of $A(t)$. Assume that there are two vector potentials $A_1(t), A_2(t)\in C_b\cap C^{2,\alpha}(\R^3;\R^3)$ satisfying (\ref{A In RVD})-(\ref{Velocity In RVD}) for fixed $f(t)\in C^{1,\alpha}_0(\R^6;\R)$. Denote by $B_R\times B_R$ the support of $f(t)$, where $B_R\subset\R^3$ is the ball centered at the origin with radius $R>0$. From Lemma \ref{Darwin Potentials Lemma}, $A_1(t)$ and $A_2(t)$ satisfy the equations
 \begin{equation}\label{eqn add1}
\Delta A_i(x) = -4\pi j_{A_i}(x)- \nabla\left\{\nabla\cdot\int_{\mathbb{R}^3}j_{A_i}(y)\frac{dy}{\left|y-x\right|}\right\}, \quad \nabla\cdot A_i(x)=0, 
\end{equation}
for $i=1, 2$, where 
\[j_{A_i}(x)=\int_{\R^3} v_{A_i}(x)f(x,p)dp, \quad v_{A_i}(x)=\frac{p-A_i(x)}{\sqrt{1+\left|p-A_i(x)\right|^2}}.\]
Then,
\[-\Delta (A_1-A_2)(x) = 4\pi \left(j_{A_1}-j_{A_2}\right)(x) + \nabla\left\{\nabla\cdot\int_{\mathbb{R}^3} \left(j_{A_1}-j_{A_2}\right)(y)\frac{dy}{\left|y-x\right|}\right\}.\]
Integrating the above equation against $(A_1-A_2)(x)$ over $x\in\R^3$, we have after using an integration by parts,
\begin{eqnarray}\label{eqn add2}
\int_{\R^3} \left|\partial_x(A_1 - A_2)(x)\right|^2 dx  &=& 4\pi \int_{\R^3} (A_1-A_2)(x)\cdot (j_{A_1}-j_{A_2})(x)\,dx \nonumber\\
& & \quad + \int_{\R^3}\nabla\cdot(A_1-A_2)(x)\nabla\cdot I(x)\,dx
\end{eqnarray}
where 
\[ I(x) = \int_{\mathbb{R}^3} \left(j_{A_1}-j_{A_2}\right)(y)\frac{dy}{\left|y-x\right|}.\]
Notice that the boundary terms in the integrations vanish. Indeed, $\partial_x(A_1-A_2)(x)$ and $\partial_x I(x)$ have a decay $O(\left|x\right|^{-2})$, and $(A_1-A_2)(x)$ has a decay $O(\left|x\right|^{-1})$. Then the products $(A_1-A_2)\partial_xI(x)$ and $(A_1-A_2)\partial_x(A_1-A_2)(x)$ decay like $O(\left|x\right|^{-3})$, which is sufficient for the disappearance of boundary terms. Since $\nabla\cdot (A_1-A_2)(x)=0$, then the last integral term in (\ref{eqn add2}) also vanishes. Hence (\ref{eqn add2}) reads as:
\begin{eqnarray}\label{eqn add3}
\lefteqn{\int_{\R^3} \left|\partial_x(A_1 - A_2)(x)\right|^2 dx}  \nonumber \\
& &  -  4\pi \int_{B_R}\int_{B_R} \left[(v_{A_1}-v_{A_2})\cdot (A_1-A_2)\right](x,p) f(x,p)\, dp\, dx = 0.
\end{eqnarray}
Now consider the vector-valued function $v(z) := \frac{z}{\sqrt{1+|z|^2}}, \; z\in\R^3$. Clearly, $v\in C_b^1(\R^3;\R^3)$, and its matrix derivative is
\[Dv(z)=\frac{1}{\sqrt{1+|z|^2}}\left[\texttt{id}-\frac{z\otimes z}{1+|z|^2}\right].\]
It is easy to check that $Dv(z)$ is symmetric, positive definite, with determinant $\texttt{det}Dv(z) =(1+|z|^2)^{-5/2}$. Moreover, we can write $v_{A_i}$ in the form $v_{A_i}=v(g_{A_i})$ where $g_{A_i}(x,p)=p-A_i(x)$. Then, using the mean value theorem on $v$, and the fact that $(x,p)\in B_R\times B_R$ and $A_i\in C_b(\R^3;\R^3)$, we have for all $(x,p)\in B_R\times B_R$,
\begin{eqnarray}\label{eqn add4}
- (v_{A_1}-v_{A_2})\cdot (A_1-A_2) &=& - \left(v(g_{A_1})-v(g_{A_2})\right)\cdot (A_1-A_2) \nonumber\\
&=& Dv(g)(A_1-A_2)\cdot (A_1-A_2)\nonumber\\
& \geq& \lambda |A_1-A_2|^2,
\end{eqnarray}
where $\lambda>0$ can be chosen uniformly in $(x,p)\in B_R\times B_R$. Here $g$  is of the form $g=g_{A_2}+\theta(g_{A_1}-g_{A_2})$ for some $\theta\in [0,1]$. Inserting (\ref{eqn add4}) into (\ref{eqn add3}), we have:
\[\int_{\R^3} \left|\partial_x(A_1 - A_2)(x)\right|^2 dx   +  4\pi\lambda \int_{B_R} |(A_1-A_2)(x)|^2 \rho(x)\, dx \leq 0.\]
Since the left-hand side is non-negative, we then deduce that $A_1=A_2$, and therefore uniqueness. 
\end{proof}

%%%%%%%%%%%%%%%%%%%%%%%%%%%%%%%%%%%%%%%

\begin{remark}
\label{Regularity In t Remark}
 If we consider the time dependence in Lemma \ref{Fixed Point Lemma}, and assume that $f$ is $C^1$ with respect to $t\in I$, then $A$ is also $C^1$ in $t\in I$ as a consequence of the Implicit Function Theorem in Banach spaces; see \cite{Slezak}.
\end{remark}

%%%%%%%%%%%%%%%%%%%%%%%%%%%%%%%%%%%%%%%%%%%%%%%%%%%%%%%%%%  5 %%%%%%%%%%%%%%%%%%%%%%%%%%%%%%%%%%%%%%%%
%%%%%%%%%%%%%%%%%%%%%%%%%%%%%%%%%%%%%%%%%%%%%%%%%%%%%%%%%%%%%%%%%%%%%%%%%%%%%%%%%%%%%%%%%%%%%%%%%%%%%%

\subsection{Estimates on $\partial_tA$ and its space derivative}
\label{Subsection Time Derivatives A}

We now turn to the estimates on the time derivative of the vector potential (i.e., the transversal component of the electric field), and its space derivatives. 

Throughout this section, we shall assume that the triplet $(f,\Phi,A)$ satisfies (\ref{Vlasov In RVD})-(\ref{Velocity In RVD}) on $I\times\mathbb{R}^3\times\mathbb{R}^3$ with $f(t)$ having compact support on $\mathbb{R}^3\times\mathbb{R}^3$. For $f$ as given, define $t\mapsto\bar{Z}(t)$ by
\begin{equation}
\label{Support Momenta}
\bar{Z}(t) = \sup\left\{\left|(x,p)\right|:\exists0\leq s\leq t:f(s,x,p)\neq0\right\}.
\end{equation}
The function $\bar{Z}(t)$ is a non-decreasing function of $t$, which by the compact support of $f$ is bounded on any finite subinterval of $I$. The following lemma is essential to our results:

\begin{lemma} 
\label{L2 Estimate Time Derivative A}
Let $f\in C^1(I,C^{1,\alpha}_0(\mathbb{R}^6);\mathbb{R})$ and $(\Phi,A)\in C^1(I,C^{2,\alpha}(\mathbb{R}^3);\mathbb{R}^3\times\mathbb{R}^3)$, with $f\geq0$ and $0<\alpha<1$, satisfy (\ref{Vlasov In RVD})-(\ref{Velocity In RVD}). Define $\rho$ and $\bar{Z}(t)$ according to (\ref{Density and Current}) and (\ref{Support Momenta}), respectively. There exists a positive $C(t)=C(\bar{Z}(t),\left\|f(t)\right\|_{L^{\infty}_{x,p}})$ such that
$$\left\|\partial_t\partial_xA(t)\right\|_{L^2_x}+\left\|\rho^{1/2}(t)\partial_tA(t)\right\|_{L^2_x}\leq C(t),\quad t\in I.$$
\end{lemma}
\begin{remark} 
\label{Remark Velocity Control}
For $t\in I$, Lemma \ref{Estimate Darwin  Potentials} and the assumption on the support of $f$ imply $\left|p-A\right|\leq C(\bar{Z}(t),\left\|f(t)\right\|_{L^{\infty}_{x,p}})<\infty$ and so $\left|v_A\right|<1$ strictly on $\texttt{supp}f(t)$.
\end{remark}

\begin{proof}[Proof of Lemma \ref{L2 Estimate Time Derivative A}]
  For $v_A$ as given in (\ref{Velocity In RVD}) define the current density 
      $$j_A(t,x)=\int_{\mathbb{R}^3}v_Af(t,x,p)dp.$$
By Lemma \ref{Darwin Potentials Lemma}(c), the components $A^i$, $i=1,2,3$ of the vector potential satisfy
\begin{equation}
\label{Components}
\Delta A^i(t,x)=-4\pi j^i_A(t,x)-\partial_{x_i}\int_{\mathbb{R}^3}\nabla\cdot j_A(t,y)\frac{dy}{\left|y-x\right|}.
\end{equation}
Take the partial time derivative on both sides of the above equation and multiply by $\partial_tA^i$. After integration by parts, dropping the $4\pi$ and using the definition of $j_A$, we have
\begin{eqnarray} 
\label{Components By Parts}
\int_{\mathbb{R}^3}\left|\nabla\partial_tA^i\right|^2(t,x)dx & = & \int_{\mathbb{R}^3}\int_{\mathbb{R}^3}\partial_tA^i(t,x)\partial_t\left(v^i_Af\right)(t,x,p)dpdx\nonumber\\
& &-\int_{\mathbb{R}^3}\int_{\mathbb{R}^3}\partial_t\partial_{x_i}A^i(t,x)\nabla\cdot \partial_tj_A(t,y)\frac{dxdy}{\left|y-x\right|},
\end{eqnarray}
where $\partial_tA^i(t,x)\partial_t\left(v^i_Af\right)$ simply means here the product of the two terms, not their sum, contrarily to the repeated index summation notation used throughout the paper. Note that the boundary terms vanish. Indeed, since $f$ has a compact support, so does $j_A$ and the boundary term corresponding to the first term on the right-hand side of the above equation is zero. On the other hand, it follows by standard arguments that $\partial_tA^i(x)$ has at least a decay $O(\left|x\right|^{-1})$ and $\nabla\partial_tA^i(x)=O(\left|x\right|^{-2})$. Moreover, the integral $I(x)$ on the right-hand side of (\ref{Components}) has a decay $O(\left|x\right|^{-2})$ and so does $\partial_tI(x)$. Therefore, $\partial_tA^i\nabla\partial_tA^i(x)=O(\left|x\right|^{-3})$ and $\partial_tA^i\partial_tI(x)=O(\left|x\right|^{-3})$, which suffice for the boundary terms to vanish. Now we add the equations (\ref{Components By Parts}) for each component of $A$. We find
\begin{eqnarray}
\label{Working Integrals}
\int_{\mathbb{R}^3}\left|\partial_x\partial_tA\right|^2 & = & \int_{\mathbb{R}^3}\int_{\mathbb{R}^3}f\left(\partial_tA\cdot\partial_tv_A\right)+\int_{\mathbb{R}^3}\int_{\mathbb{R}^3} \left(\partial_tA\cdot v_A\right)\partial_tf\nonumber\\
& & -\partial_t\int_{\mathbb{R}^3}\int_{\mathbb{R}^3}\frac{1}{r}\left(\nabla\cdot A\right)\left(\nabla\cdot\partial_tj_A\right)\nonumber\\
& = & I_1+I_2+I_3.
\end{eqnarray}
But $I_3\equiv0$ since the vector potential satisfies the Coulomb gauge condition; see Lemma \ref{Darwin Potentials Lemma}(c). Also, by using the representation of the derivatives of the velocity given in the Appendix, the integral $I_1$ can be written as
\begin{eqnarray}
   I_1 & = & -\int_{\mathbb{R}^3}\int_{\mathbb{R}^3}\frac{f}{\sqrt{1+g^2}}\left(\left|\partial_tA\right|^2-\left|v_A\cdot\partial_tA\right|^2\right),\nonumber
  \end{eqnarray}
where we have denoted $g=\left|p-A\right|$. We shall also denote $K=-\nabla\Phi+v^i_A\nabla A^i.$ Hence, after sending $I_1$ to the left-hand side of (\ref{Working Integrals}), and by using the Vlasov equation (\ref{Vlasov In RVD}) in $I_2$, we find that 
\begin{eqnarray}
\label{LHS and RHS}
\int_{\mathbb{R}^3}\left|\partial_x\partial_tA\right|^2 & + & \int_{\mathbb{R}^3}\int_{\mathbb{R}^3}\frac{f}{\sqrt{1+g^2}}\left(\left|\partial_tA\right|^2-\left|v_A\cdot\partial_tA\right|^2\right) 
\nonumber\\
& = & -\int_{\mathbb{R}^3}\int_{\mathbb{R}^3}\left(v_A\cdot\partial_tA\right)\left[\nabla_x\cdot\left(v_Af\right)+\nabla_p\cdot\left(Kf\right)\right]\nonumber\\
& = & \int_{\mathbb{R}^3}\int_{\mathbb{R}^3}f\partial_tA^i\left(v_A\cdot\nabla_xv^i_A\right)+\int_{\mathbb{R}^3}\int_{\mathbb{R}^3}fv^i_A\left(v_A\cdot\nabla_x\partial_tA^i\right)\nonumber\\
   & & + \int_{\mathbb{R}^3}\int_{\mathbb{R}^3}f\partial_tA^i\left(K\cdot\nabla_pv^i_A\right).
\end{eqnarray}
Notice the integration by parts and the use of the product rule in the last equality. We claim that the left-hand side of the above equality has a positive lower bound for every time $t$. Indeed, we have that 
\begin{equation}\label{inequality on A}
\left|\partial_tA\right|^2-\left|v_A\cdot\partial_tA\right|^2\geq\left|\partial_tA\right|^2-\left|v_A\right|^2\left|\partial_tA\right|^2=\left(1-\left|v_A\right|^2\right)\left|\partial_tA\right|^2.
\end{equation}
Also, by Remark \ref{Remark Velocity Control} there exists a $g_{\texttt{max}}(t)=g_{\texttt{max}}(\bar{Z}(t),\left\|f(t)\right\|_{L^{\infty}_{x,p}})<\infty$ so that  
\begin{equation}
\label{Lower Bound} \frac{1-\left|v_A\right|^2}{\sqrt{1+g^2}}\equiv\frac{1}{\left(1+g^2\right)^{3/2}}\geq\frac{1}{\left(1+g^2_{\texttt{max}}\right)^{3/2}}\equiv C_{\texttt{min}}>0.
\end{equation}
Therefore, the left-hand side of (\ref{LHS and RHS}) satisfies
\begin{equation}
\label{LHS}
 \texttt{LHS}\geq C_{\texttt{min}}(t)\left(\left\|\partial_x\partial_tA(t)\right\|^2_{L^2_x}+\left\|\rho^{1/2}(t)\partial_tA(t)\right\|^2_{L^2_x}\right).
\end{equation}
On the other hand, the known bounds on the derivatives of the potentials given in Lemma \ref{Estimate Darwin Potentials} lead to  
$$\left\|\partial_xv_A(t)\right\|_{L^{\infty}_{x,p}}+\left\|\partial_pv_A(t)\right\|_{L^{\infty}_{x,p}}+\left\|K(t)\right\|_{L^{\infty}_{x,p}}\leq C(t)\equiv C(\bar{Z}(t),\left\|f(t)\right\|_{L^{\infty}_{x,p}});$$
see the Appendix for an explicit representation of the derivatives of the velocity. Hence, after a use of the Cauchy-Schwarz inequality and again the use of the compact support of $f$, the right-hand side of (\ref{LHS and RHS}) can be estimated as 
\begin{equation}
\label{RHS}
\texttt{RHS}\leq C(t)\left(\left\|\partial_x\partial_tA(t)\right\|_{L^2_x}+\left\|\rho^{1/2}(t)\partial_tA(t)\right\|_{L^2_x}\right).
\end{equation}
Finally, since $\left(a+b\right)^2\leq 2\left(a^2+b^2\right)$, the result follows from (\ref{LHS})-(\ref{RHS}).
\end{proof}

\begin{lemma} 
\label{LInfinity Estimate Time Derivative A}
Under the assumptions of Lemma \ref{L2 Estimate Time Derivative A}, we have that 
\begin{eqnarray}
\left\|\partial_tA(t)\right\|_{L^{\infty}_x} & \leq & C\left[\left\|\rho(t)\right\|^{1/3}_{L^1_x}\left\|\rho(t)\right\|^{2/3}_{L^{\infty}_x}\left(1+\left\|\rho(t)\right\|^{2/3}_{L^1_x}\left\|\rho(t)\right\|^{1/3}_{L^{\infty}_x}\right)\right.\nonumber\\
& & \left.+ \left\|\rho(t)\right\|^{1/6}_{L^1_x}\left\|\rho(t)\right\|^{1/3}_{L^{\infty}_x}\left\|\rho^{1/2}(t)\partial_tA(t)\right\|_{L^2_x}\right],\quad t\in I.\hspace{-.5cm}\nonumber
\end{eqnarray}
\end{lemma}

\begin{corollary} 
\label{Corollary Estimate Time Derivative}
Under the assumptions of Lemma \ref{L2 Estimate Time Derivative A}, we have that 
$$\left\|\partial_tA(t)\right\|_{L^{\infty}_x}\leq C(t),\quad t\in I,$$
for some positive $C(t)=C(\bar{Z}(t),\left\|f(t)\right\|_{L^{\infty}_{x,p}})$.
\end{corollary}

\begin{proof}[Proof of Lemma \ref{LInfinity Estimate Time Derivative A}]
Consider the integral representation (\ref{A In RVD}) for the vector potential, and take the partial time derivative on both sides of this equation. Denoting the kernel by $\mathcal{K}(x,y)$ and dropping the multiple $1/2$, we have
  \begin{eqnarray}
     \partial_tA(t,x) & = & \int_{\mathbb{R}^3}\int_{\mathbb{R}^3}\mathcal{K}(x,y)\left[v_A\partial_tf+\partial_tv_Af\right](t,y,p)dpdy\nonumber\\
    & = & I_1+I_2.\nonumber   
  \end{eqnarray}
Set $K=-\nabla\Phi+v^i_A\nabla A^i$. A use of the Vlasov equation yields
\begin{eqnarray}
  I_1 & = & -\int_{\mathbb{R}^3}\int_{\mathbb{R}^3}\mathcal{K}(x,y)v_A\nabla_y\cdot\left(v_Af\right)-\int_{\mathbb{R}^3}\int_{\mathbb{R}^3}\mathcal{K}(x,y)v_A\nabla_p\cdot\left(Kf\right)\nonumber\\
   & = &  \int_{\mathbb{R}^3}\int_{\mathbb{R}^3}fv_A\cdot\left[\partial_y\mathcal{K}(x,y)v_A+\mathcal{K}(x,y)\partial_yv_A\right]+\int_{\mathbb{R}^3}\int_{\mathbb{R}^3}fK\cdot\mathcal{K}(x,y)\partial_pv_A.\nonumber
\end{eqnarray}
Therefore, since $\left|v_A\right|\leq 1$, also $\left|\partial_x\mathcal{K}(x,y)\right|\leq C\left|y-x\right|^2$ (see Lemma \ref{Darwin Potentials Lemma}(c)), and $\left|\partial_xv_A\right|\leq C\left|\partial_xA\right|$ and $\left|\partial_pv_A\right|\leq C$ (see Appendix), we have
\begin{eqnarray}
\label{Estimate I1 dtA}
  I_1 & \leq & C\left\{\int_{\mathbb{R}^3}\rho(t,y)\frac{dy}{\left|y-x\right|^2}+\left(\left\|\partial_x\Phi(t)\right\|_{L^{\infty}_x}+\left\|\partial_xA(t)\right\|_{L^{\infty}_x}\right)\int_{\mathbb{R}^3}\rho(t,y)\frac{dy}{\left|y-x\right|}\right\}\nonumber \hspace{-.5cm}\\
  & \leq & C\left\|\rho(t)\right\|^{1/3}_{L^1_x}\left\|\rho(t)\right\|^{2/3}_{L^{\infty}_x}\left(1+\left\|\rho(t)\right\|^{2/3}_{L^1_x}\left\|\rho(t)\right\|^{1/3}_{L^{\infty}_x}\right),
\end{eqnarray}
where in the last inequality we used the estimates from Lemmas \ref{Pallard} and \ref{Estimate Darwin Potentials}. 

On the other hand, since $\left|\partial_tv_A\right|\leq C\left|\partial_tA\right|$ (see Appendix), the integral $I_2$ can be estimated as
  \begin{eqnarray}
  \label{Does Not Allow}
     I_2 & \leq &  C\int_{\mathbb{R}^3}\rho(t,y)\left|\partial_tA(t,y)\right|\frac{dy}{\left|y-x\right|}\nonumber\\
     & \leq & C\left\|\rho(t)\partial_tA(t)\right\|^{1/3}_{L^1_x}\left\|\rho(t)\partial_tA(t)\right\|^{2/3}_{L^2_x},\nonumber   
  \end{eqnarray}
where the second inequality is a consequence of Lemma \ref{Pallard}. Hence, the Cauchy-Schwarz inequality and a direct estimate lead to
  \begin{eqnarray}
  \label{Estimate I2 dtA}
     I_2 & \leq &  C\left\|\rho(t)\right\|^{1/6}_{L^1_x}\left\|\rho(t)\right\|^{1/3}_{L^{\infty}_x}\left\|\rho^{1/2}(t)\partial_tA(t)\right\|_{L^2_x}.
  \end{eqnarray}
The lemma then follows from (\ref{Estimate I1 dtA}) and (\ref{Estimate I2 dtA}).
\end{proof}

\begin{lemma}
\label{LInfinity Estimate Double Time Derivative A} 
Under the assumptions of Lemma \ref{L2 Estimate Time Derivative A}, we also have 
$$
\left\|\partial_t\partial_xA(t)\right\|_{L^{\infty}_x} \leq  C(t)\left(\left\|\partial_tf(t)\right\|_{L^{\infty}_{x,p}}+\left\|\rho(t)\right\|^{1/3}_{L^1_x}\left\|\rho(t)\right\|^{2/3}_{L^{\infty}_x}\left\|\partial_tA(t)\right\|_{L^{\infty}_x}\right), t\in I.
$$
for some $C(t)=C(\bar{Z}(t))$.
\end{lemma}
\begin{corollary} 
\label{Corollary Estimate Double Time Derivative}
Under the assumptions of Lemma \ref{L2 Estimate Time Derivative A}, we have that 
$$\left\|\partial_t\partial_xA(t)\right\|_{L^{\infty}_x}\leq C(t),\quad t\in I,$$
for some positive $C(t)=C(\bar{Z}(t),\left\|f(t)\right\|_{L^{\infty}_{x,p}},\left\|\partial_tf(t)\right\|_{L^{\infty}_{x,p}})$.
\end{corollary}

\begin{proof}[Proof of Lemma \ref{LInfinity Estimate Double Time Derivative A}]
 Consider
  \begin{eqnarray}
    \label{Here And Now Useful}
     \partial_t\partial_xA(t,x) & = & \int_{\mathbb{R}^3}\int_{\mathbb{R}^3}\partial_x\mathcal{K}(x,y)\left[v_A\partial_tf+\partial_tv_Af\right](t,y,p)dpdy\nonumber\\
     & = & I_1+I_2,   
  \end{eqnarray}
where $\left|\partial_x\mathcal{K}(x,y)\right|\leq C\left|y-x\right|^{-2}$. Hence, by Lemma \ref{Pallard}, we obtain  
$$I_1\leq C(\bar{Z}(t))\left\|\partial_tf(t)\right\|_{L^{\infty}_{x,p}}\quad\hbox{and}\quad I_2\leq C\left\|\partial_tA(t)\right\|_{L^{\infty}_x}\left\|\rho(t)\right\|^{1/3}_{L^1_x}\left\|\rho(t)\right\|^{2/3}_{L^{\infty}_x},$$
where $\left|\partial_tv_A\right|\leq C\left|\partial_tA\right|$ has been used. The result readily follows.
\end{proof}

%%%%%%%%%%%%%%%%%%%%%%%%%%%%%%%%%%%%%%%%%%%%%%%%%%%%%%%%%%  6 %%%%%%%%%%%%%%%%%%%%%%%%%%%%%%%%%%%%%%%%%%%%%%
%%%%%%%%%%%%%%%%%%%%%%%%%%%%%%%%%%%%%%%%%%%%%%%%%%%%%%%%%%%%%%%%%%%%%%%%%%%%%%%%%%%%%%%%%%%%%%%%%%%%%%%%%%%%

\section{The Cauchy Problem for the RVD System}
\label{Section Small Data Solutions}

A noticeable advantage of writing the RVD system in terms of the generalized variables and potentials is that it resembles, to some extent, the well-known Vlasov-Poisson (VP) system. Actually, the latter can be formally obtained from (\ref{Vlasov In RVD})-(\ref{Velocity In RVD}) by letting $c\rightarrow\infty$, so that terms involving the vector potential are no longer present. This resemblance allows to adapt previous techniques used for the VP system to the Darwin case. Below, the proofs we present are in the same vein as those given in \cite{Rein} for the VP system. Obviously, several non-trivial difficulties arise due to the inclusion of the vector potential in the system equations, not present in the Poisson case. Incidentally, we expect that a global in time existence result to the \textit{relativistic} Vlasov-Poisson system for unrestricted Cauchy data, which is still unsolved, will lead to an analogous result for the RVD system.   

%%%%%%%%%%%%%%%%

\subsection{Local Solutions} 
\label{Local Solutions Section}

In this section we shall produce a local in time existence and uniqueness result for classical solutions of the RVD system.  

\begin{definition}
\label{Classical Solution RVD Definition}
Let $f_0$ be given. We call $f$ a classical solution of the RVD system if $f\in C^1(I\times\mathbb{R}^6;\mathbb{R})$; it induces the potentials $(\Phi,A)\in C^1(I,C^2(\mathbb{R}^3);\mathbb{R}\times\mathbb{R}^3)$ via (\ref{Phi In RVD})-(\ref{A In RVD}); for every compact interval $\bar{J}\subset I$ the fields $\nabla\Phi$ and $v^i_A\nabla A^i$ are bounded on $\bar{J}\times\mathbb{R}^3$ and $\bar{J}\times\mathbb{R}^3\times\mathbb{R}^3$ respectively; and the triplet $(f,\Phi,A)$ satisfies the system (\ref{Vlasov In RVD})-(\ref{Velocity In RVD}) on $I\times\mathbb{R}^3\times\mathbb{R}^3$. Moreover, we say that $f$ is a classical solution of the Cauchy problem if $\left.f\right|_{t=0}=f_0$. 
\end{definition}
 
\begin{theorem}
\label{Local Solutions Theorem}
Let $f_0\in C^{1,\alpha}_0(\mathbb{R}^6;\mathbb{R})$, $0<\alpha<1$, $f_0\geq0$. For some $T>0$, there exists a unique classical solution $f$  on $[0,T[$ of the Cauchy problem for the RVD system. Moreover, for each $0\leq t<T$, the function $f(t)$ is in $C^{1,\alpha}(\mathbb{R}^6;\mathbb{R})$, it is non-negative, and has compact support. In addition, if $T>0$ is the life span of $f$,  then
$$ \sup\left\{\left|p\right|:\exists0\leq t<T,x\in\mathbb{R}^3:f(t,x,p)\neq0\right\}<\infty
$$
implies that the solution is global in time, i.e., $T=\infty$.
\end{theorem}

\paragraph{\textit{\textbf{Uniqueness.}}} Consider two solutions $(f_1,\Phi_1,A_1)$ and $(f_2,\Phi_2,A_2)$ of the RVD system as given by Theorem \ref{Local Solutions Theorem}. Then for $i=1,2$, there exists $T_i>0$ such that $f_i(t)\in C^{1,\alpha}(\R^6;\R)$, with support in the ball $B_{R_i}\times B_{R_i}$, uniformly in $t\in (0,T_i)$. Then setting $R=max(R_1,R_2)$ and $T=min(T_1,T_2)$, we have that 
$$ \texttt{supp}f_1(t)\cup\texttt{supp}f_2(t)\subset B_R\times B_R,\quad t\in[0,\bar{T}]\subset[0,T[. $$
The Vlasov equation yields 
\begin{eqnarray}  \partial_t\left(f_1-f_2\right)^2 & + & v_{A_1}\cdot\nabla_x\left(f_1-f_2\right)^2+K_1\cdot\nabla_p\left(f_1-f_2\right)^2\nonumber\\
& = & 2\left(f_1-f_2\right)\left[\left(v_{A_2}-v_{A_1}\right)\cdot\nabla_xf_2+\left(K_2-K_1\right)\cdot\nabla_pf_2\right]\nonumber
\end{eqnarray}  
where $K_1=-\nabla\Phi_1+v^i_{A_1}\nabla A^i_1$ and analogously for $K_2$. Now, let $Q(t)=\left\|f_1(t)-f_2(t)\right\|^2_{L^2_{x,p}}$.  Since $f_2\in C^1([0,T[\times\mathbb{R}^6;\mathbb{R})$ has compact support,  \\ $\left\|\nabla_xf_2(t)\right\|_{L^{\infty}_{x,p}}+\left\|\nabla_pf_2(t)\right\|_{L^{\infty}_{x,p}}\leq C_R$. Also, we have $\left|v_{A_1}-v_{A_2}\right|\leq C\left|A_1-A_2\right|$ and, by Lemma \ref{Estimate Darwin Potentials}, $\left\|\partial_x(A_1,A_2)\right\|_{L^{\infty}_x}\leq C_R$. Then, it is not difficult to check that
\begin{eqnarray}
\label{Gronwall}
\frac{dQ(t)}{dt}& \leq & C_RQ^{1/2}(t)\left[\frac{}{}\left\|\partial_x\Phi_1(t)-\partial_x\Phi_2(t)\right\|_{L^2_x(B_R)}\right.\nonumber\\
& & \left.+\left\|A_1(t)-A_2(t)\right\|_{L^2_x(B_R)}
+\left\|\partial_x A_1(t)-\partial_x A_2(t)\right\|_{L^2_x(B_R)}\right].
\end{eqnarray}
From the Poisson equation satisfied by the scalar potentials we deduce
$$ \int_{\mathbb{R}^3}\left|\partial_x\Phi_1(t,x)\right|^2dx=4\pi\int_{\mathbb{R}^3}\int_{\mathbb{R}^3}\rho_1(t,x)\rho_1(t,y)\frac{dxdy}{\left|y-x\right|},
$$
and analogously for $\Phi_2$. Linearity, the Hardy-Littlewood-Sobolev inequality and Jensen's inequality yield
\begin{equation}
\label{Estimate For Gronwall 1}
 \left\|\partial_x\Phi_1(t)-\partial_x\Phi_2(t)\right\|_{L^2_x}\leq C \left\|\rho_1(t)-\rho_2(t)\right\|_{L^{6/5}_x}\leq C_R Q^{1/2}(t).
\end{equation}
On the other hand, in order to estimate the terms involving the vector potential, we proceed as follows. Define $f_{\lambda}=\lambda f_1+\left(1-\lambda\right)f_2$ for $0\leq\lambda\leq1$. Clearly $f_{\lambda}\geq 0$ has compact support and satisfies $\partial_{\lambda}f_{\lambda}=f_1-f_2$. Let $A_{\lambda}$ be the Darwin vector potential induced by $f_{\lambda}$. In view of Lemma \ref{Darwin Potentials Lemma} we have
$$ \Delta A_{\lambda}(t,x)=-4\pi\int_{B_R}v_{A_{\lambda}}f_{\lambda}(t,x,p)dp-\nabla\int_{B_R}\int_{B_R}\nabla\cdot\left(v_{A_{\lambda}}f_{\lambda}\right)(t,y,p)\frac{dpdy}{\left|y-x\right|},
$$
and $\nabla\cdot A_{\lambda}=0$. Notice that $A_1$ (resp. $A_2$) solves the above equation when $\lambda=1$ (resp. $\lambda=0$). By virtue of Remark \ref{Regularity In t Remark} (where $t$ is replaced by $\lambda$), we can use the arguments in the proof of Lemma \ref{L2 Estimate Time Derivative A} to find 
\begin{eqnarray}
\label{Analogous To Lemma}
\int_{\mathbb{R}^3}\left|\partial_{\lambda}\partial_xA_{\lambda}\right|^2 & + & \int_{B_R}\int_{B_R}\frac{f_{\lambda}}{\sqrt{1+g_{\lambda}}}\left(\left|\partial_{\lambda}A_{\lambda}\right|^2-\left|v_{A_{\lambda}}\cdot\partial_{\lambda}A_{\lambda}\right|^2\right)\nonumber\\
& = & \int_{B_R}\int_{B_R}\left(\partial_{\lambda}A_{\lambda}\cdot v_{A_{\lambda}}\right)\partial_{\lambda}f_{\lambda}.
\end{eqnarray}
The analogous expression in Lemma \ref{L2 Estimate Time Derivative A} is the first equality in (\ref{LHS and RHS})  with the right-hand side replaced by the expression of $I_2$ in (\ref{Working Integrals}).
 Note the integration over $B_R\times B_R$ in view of the compact support of $\partial_{\lambda}f_{\lambda}$. Hence, since $\left|v_{A_{\lambda}}\right|<1$ by Remark \ref{Remark Velocity Control}, we can use again the arguments in Lemma \ref{L2 Estimate Time Derivative A} and the Cauchy-Schwarz inequality on the right-hand side of  (\ref{Analogous To Lemma}) to obtain
 \[ \left\|\partial_{\lambda}\partial_xA_{\lambda}(t)\right\|^2_{L^2_x}  + \left\|\rho_\lambda^{1/2}(t)\partial_{\lambda}A_{\lambda}(t)\right\|^2_{L^2_x} \leq  C_R\left\|\partial_{\lambda}A_{\lambda}(t)\right\|_{L^2_x(B_R)} \|\partial_\lambda f_\lambda(t)\|_{L^2_{x,p}},\]
 which implies that
\begin{equation}
\label{Estimate Needed 1}
 \left\|\partial_{\lambda}\partial_xA_{\lambda}(t)\right\|^2_{L^2_x}\leq  C_R\left\|\partial_{\lambda}A_{\lambda}(t)\right\|_{L^2_x(B_R)}Q^{1/2}(t)
\end{equation}
for some $C_R>0$ and all $0\leq\lambda\leq1$. Poincar{\'e}'s inequality and (\ref{Estimate Needed 1}) then yield 
\[\left\|\partial_{\lambda}A_{\lambda}(t)\right\|^2_{L^2_x(B_R)}\leq  \tilde{C}_R\left\|\partial_x\partial_{\lambda}A_{\lambda}(t)\right\|^2_{L^2_x}\leq C_R\left\|\partial_{\lambda}A_{\lambda}(t)\right\|_{L^2_x(B_R)}Q^{1/2}(t),\]
and thus, 
\begin{equation}
\label{Estimate Needed 2}
\left\|\partial_{\lambda}A_{\lambda}(t)\right\|_{L^2_x(B_R)}\leq  C_R Q^{1/2}(t)
\end{equation}
for all $0\leq\lambda\leq1$. Inserting (\ref{Estimate Needed 2}) into (\ref{Estimate Needed 1}), we also have for all $0\leq\lambda\leq1$
\begin{equation}\label{new estimate needed}
\left\|\partial_{\lambda}\partial_xA_{\lambda}(t)\right\|_{L^2_x}\leq C_RQ^{1/2}(t).
\end{equation}

Now, we observe that by Jensen's inequality,
\begin{eqnarray*}
\int_{\mathbb{R}^3}\left|\partial_xA_1(t,x)-\partial_xA_2(t,x)\right|^2dx & = & \int_{\mathbb{R}^3}\left|\int^1_0\partial_{\lambda}\partial_xA_{\lambda}(t,x)d\lambda\right|^2dx \\
& \leq & \int^1_0\int_{\mathbb{R}^3}\left|\partial_{\lambda}\partial_xA_{\lambda}(t,x)\right|^2dxd\lambda\\
& \leq & \sup_{0\leq\lambda\leq1}\int_{\mathbb{R}^3}\left|\partial_{\lambda}\partial_xA_{\lambda}(t,x)\right|^2dx,
\end{eqnarray*}
and similarly for $\left\|A_1(t)-A_2(t)\right\|_{L^2_x}$.
Then, we use (\ref{Estimate Needed 2}) and (\ref{new estimate needed}) to derive the estimate
\begin{equation}
\label{Estimate For Gronwall 2}
\left\|A_1(t)-A_2(t)\right\|_{L^2_x(B_R)}+\left\|\partial_xA_1(t)-\partial_xA_2(t)\right\|_{L^2_x}\leq C_RQ^{1/2}(t).
\end{equation}
Finally, we combine (\ref{Gronwall}), (\ref{Estimate For Gronwall 1}) and (\ref{Estimate For Gronwall 2}) to conclude that 
\[\frac{dQ(t)}{dt}  \leq  C_RQ(t).\]
Uniqueness then follows as a trivial consequence of Gronwall's lemma.\qed

\smallskip

\paragraph{\textit{\textbf{Proof of Theorem \ref{Local Solutions Theorem}}}} 
Let $f_0\in C^{1,\alpha}_0(\mathbb{R}^6;\mathbb{R})$, $f_0\geq0$. Fix $\bar{X}_0>0$ and $\bar{P}_0>0$ such that $ f_0(x,p)=0$ for $\left|x\right|>\bar{X}_0$ or  $\left|p\right|>\bar{P}_0$. We introduce the following iterative scheme. For $t\in I$ and $z=(x,p)\in\mathbb{R}^3\times\mathbb{R}^3$, define $$f^0(t,z)=f_0(z).$$
For $n\in\mathbb{N}$, assume that $f^n:I\times\mathbb{R}^6\rightarrow\mathbb{R}$ is given, and define
\begin{equation} \label{Integral Phin}
\Phi^n(t,x)  =  \int_{\mathbb{R}^3}\int_{\mathbb{R}^3}f^n(t,y,p)\frac{dpdy}{\left|y-x\right|}.
\end{equation}
By Lemma \ref{Fixed Point Lemma}, there exists a unique $A_n$ solution to the equation 
\begin{equation}
\label{Integral An}
A_n(t,x)  =  \frac{1}{2}\int_{\mathbb{R}^3}\int_{\mathbb{R}^3}\left[\texttt{id}+\omega\otimes\omega\right]v_{A_n}f^n(t,y,p)\frac{dpdy}{\left|y-x\right|},
\end{equation}
where $$v_{A_n} = \frac{p-A_n}{\sqrt{1+\left|p-A_n\right|^2}}.$$ Denote by $Z_n=(X_n,P_n)(s,t,z)$ the solution of the characteristic system
\begin{eqnarray}
\label{Characteristic X Iterates}
\dot{X}_n(s,t,z) & = & v_{A_n}(s,X_n(s,t,z),P_n(s,t,z))\\
\label{Characteristic P Iterates}
\dot{P}_n(s,t,z) & = & -\left[\nabla\Phi^n-v^i_{A_n}\nabla A^i_n\right](s,t,X_n(t,z),P_n(s,t,z))
\end{eqnarray}
with $Z_n(t,t,z)=z$. We define the $(n+1)$-th iterate of the distribution function by 
$$ f^{n+1}(t,z)=f_0(Z_n(0,t,z)).
$$
For convenience we shall also define the sequences 
$$\rho^n(t,x)=\int_{\mathbb{R}^3}f^n(t,x,p)dp\quad \hbox{and}\quad j_n(t,x)=\int_{\mathbb{R}^3}v_{A_n}f^n(t,x,p)dp.$$
\\
\noindent\textbf{Step 1}: In view of the Lemmas and Remarks in Section \ref{The Potential Representation Section}, and Lemma \ref{Fixed Point Lemma} in Section \ref{The RVD System Section}, the sequence $\left\{(f^n,\Phi^n,A_n)\right\}$ is well defined. In particular, $f^n\in C^1(I, C^{1,\alpha}(\mathbb{R}^6);\mathbb{R})$, $f^n\geq0$, and $(\Phi^n,A_n)\in C^1(I,C^{2,\alpha}(\mathbb{R}^3);\mathbb{R}\times\mathbb{R}^3)$. For each $n$, the regularity in \textit{time} of the potentials is the one of $f^n$. This is trivial for $\Phi^n$. As for $A^n$, see Remark \ref{Regularity In t Remark}. 

For $t\in I$ set $\bar{P}_0(t)=\bar{P}_0$ and for $n\in\mathbb{N}$ define
  \begin{eqnarray}
  \bar{P}_n(t) & = & \left\{\left|p\right|:\exists0\leq s\leq t,x\in\mathbb{R}^3:f^n(s,x,p)\neq0\right\}\nonumber\\
         & \equiv & \left\{\left|P_{n-1}(s,0,z)\right|:0\leq s\leq t,z\in \texttt{supp}f_0\right\}\nonumber 
   \end{eqnarray}
It is clear that $\texttt{supp}f^n(t)\subseteq \left\{(x,p)\in\mathbb{R}^3\times\mathbb{R}^3:\left|x\right|\leq\bar{X}_0+t,\left|p\right|\leq\bar{P}_n(t)\right\}$. Also,  
$$ \left\|f^n(t)\right\|_{L^q_z}=\left\|f_0\right\|_{L^q_z},\quad 1\leq q\leq\infty,\quad t\in I,\quad n\in\mathbb{N}
$$
and we have the estimate
\begin{equation}\label{eqn estimate rho n}
 \left\|\rho^n(t)\right\|_{L^{\infty}_x}\leq\frac{4}{3}\pi\left\|f_0\right\|_{L^{\infty}_z}\bar{P}^3_n(t).
\end{equation}
Since $\left|j_n\right|\leq\left|\rho^n\right|$, the known estimates on the potentials imply that 
$$ \left\|\Phi^n(t)\right\|_{L^{\infty}_x}+\left\|A_n(t)\right\|_{L^{\infty}_x}\leq C(f_0)\bar{P}_n(t),
$$
and finally 
\begin{equation}
\label{Estimate Step 1 Derivative Field}
 \left\|\partial_x\Phi^n(t)\right\|_{L^{\infty}_x}+\left\|\partial_xA_n(t)\right\|_{L^{\infty}_x}\leq C(f_0)\bar{P}^2_n(t).
\end{equation}
\\
\noindent\textbf{Step 2}: For some $T>0$ there is a non-negative, non-decreasing $\mathcal{P}\in C([0,T[;\mathbb{R})$ depending on the Cauchy datum only, such that for all $n\in\mathbb{N}\cup\left\{0\right\}$ and $0\leq t<T$ 
  $$\bar{P}_n(t)\leq\mathcal{P}(t).$$
Indeed, for $n\in\mathbb{N}$ the characteristic equation (\ref{Characteristic P Iterates}) and the estimate (\ref{Estimate Step 1 Derivative Field}) imply 
\begin{eqnarray}
\label{Uniform Estimate Momenta Iterates}
  \left|P_n(s,0,z)\right| & \leq & \left|p\right|+\int^s_0\left(\left\|\partial_x\Phi^n(\tau)\right\|_{L^{\infty}_x}+\left\|\partial_xA_n(\tau)\right\|_{L^{\infty}_x}\right)d\tau\nonumber\\
  & \leq & \bar{P}_0+C(f_0)\int^t_0\bar{P}^2_n(\tau)d\tau.
\end{eqnarray}
Let $T>0$ be the life span of the solution of the integral equation
\begin{equation}
\label{Equation Maximal P(t)}
 \mathcal{P}(t) = \bar{P}_0+C(f_0)\int^t_0\mathcal{P}^2(\tau)d\tau.
\end{equation}
Hence, $\bar{P}_0(t)\leq\mathcal{P}(t)$. Suppose $\bar{P}_n(t)\leq\mathcal{P}(t)$ for some $n\in\mathbb{N}$. Then, in view of (\ref{Uniform Estimate Momenta Iterates}), this estimate also holds for $\bar{P}_{n+1}(t)$, which proves the claim. As a result, all estimates in Step 1 are uniform in $n$ on any subinterval $[0,\bar{T}]\subset[0,T[$. In particular, for all $n\in\mathbb{N}$ and $0\leq t<T$, we have
\begin{equation}
\label{Estimate Step 2 Derivative Field}
\left\|\rho^n(t)\right\|_{L^{\infty}_x}+\left\|j_n(t)\right\|_{L^{\infty}_x}+\left\|\partial_x\Phi^n(t)\right\|_{L^{\infty}_x}+\left\|\partial_xA_n(t)\right\|_{L^{\infty}_x}\leq C^0_{\bar{T}}\equiv C(\bar{T},f_0).
\end{equation}
For future use, we notice that the maximal solution of (\ref{Equation Maximal P(t)}) is given by 
\begin{equation}
\label{Life Span}
  \mathcal{P}(t)= \bar{P}_0\left(1-C(f_0)\bar{P}_0t\right)^{-1}, \quad 0\leq t < T \equiv\left(C(f_0)\bar{P}_0\right)^{-1},
\end{equation}
with $ C(f_0)=3(2\pi)^{2/3}\left\|f_0\right\|^{1/3}_{L^1_{x,p}}\left\|f_0\right\|^{2/3}_{L^{\infty}_{x,p}}$.\\
\\

\noindent\textbf{Step 3}: We claim that for every fixed $0\leq\bar{T}<T$
\begin{equation}
\label{Estimate Step 3 Derivative Field}
\left\|\partial_x\rho^n(t)\right\|_{L^{\infty}_x}+\left\|\partial_xj_n(t)\right\|_{L^{\infty}_x}+\left\|\partial^2_x\Phi^n(t)\right\|_{L^{\infty}_x}+\left\|\partial^2_xA_n(t)\right\|_{L^{\infty}_x}\leq C^0_{\bar{T}},
\end{equation}
for all $n\in\mathbb{N}$ and $0\leq t\leq\bar{T}$. 

To start with, we estimate the space derivatives of the characteristic curves. To ease notation, we write $(X_n,P_n)(s)\equiv (X_n,P_n)(s,t,x,p)$. Recall  
$$ v_{A_n}(s,X_n(s),P_n(s))\equiv v(P_n(s),A_n(s,X_n(s))),
$$  
where $v$ is $C^{\infty}_b$ in its argument. Hence, since (by abuse of notation) we have $(\partial_xX_n(t),\partial_xP_n(t))=(1,0)$, the uniform bounds in (\ref{Estimate Step 2 Derivative Field}) lead to
\begin{eqnarray}
  \left|\partial_xX_n(s)\right| & \leq & \left|\partial_xX_n(t)\right|+\int^t_s\left|\partial_x\left[v(P_n(\tau),A_n(\tau,X_n(\tau)))\right]\right|d\tau\nonumber\\
  & \leq & 1 + C^0_{T}\int^t_0\left(\left|\partial_xX_n(\tau)\right|+\left|\partial_xP_n(\tau)\right|\right)d\tau.\nonumber
\end{eqnarray}  
Similarly,
\begin{eqnarray}
  \left|\partial_xP_n(s)\right| & \leq & \left|\partial_xP_n(t)\right|+\int^t_s\left|\partial_x\left[\nabla\Phi-v^i_{A_n}\nabla A^i_n\right](\tau,X_n(\tau),P_n(\tau))\right|d\tau\nonumber\\
  & \leq & C^0_{T}\int^t_0\left(1+\left\|\partial^2_x\Phi^n(\tau)\right\|_{L^{\infty}_x}+\left\|\partial^2_xA_n(\tau)\right\|_{L^{\infty}_x}\right)\nonumber\\
  & & \hspace{4.cm}\times\left(\frac{}{}\left|\partial_xX_n(\tau)\right|+\left|\partial_xP_n(\tau)\right|\frac{}{}\right)d\tau.\nonumber
\end{eqnarray}  
These two estimates and the Gronwall's lemma yield 
\begin{eqnarray}
\left|\partial_xX_n(s)\right| & + & \left|\partial_xP_n(s)\right|\nonumber\\
& \leq & \exp\left\{C^0_{T}\int^t_0\left(1+\left\|\partial^2_x\Phi^n(\tau)\right\|_{L^{\infty}_x}+\left\|\partial^2_xA_n(\tau)\right\|_{L^{\infty}_x}\right)d\tau\right\}.\nonumber
\end{eqnarray}
As a result, we also have 
\begin{eqnarray}
  \left|\partial_x\rho^{n+1}(t,x)\right| & \leq & \int_{\left|p\right|\leq\mathcal{P}(t)}\left|\partial_x\left[f_0(Z_n(0,t,x,p))\right]\right|dp\nonumber\\
  & \leq & C^0_{\bar{T}}\exp\left\{C^0_{T}\int^t_0\left(1+\left\|\partial^2_x\Phi^n(\tau)\right\|_{L^{\infty}_x}+\left\|\partial^2_xA_n(\tau)\right\|_{L^{\infty}_x}\right)d\tau\right\}.\nonumber
\end{eqnarray}  
Similarly, after using the product rule and the known estimates, there exists a sufficiently large constant $C^0_{\bar{T}}$ such that
\begin{eqnarray}
  \left|\partial_xj_{n+1}(t,x)\right| & \leq & \int_{\left|p\right|\leq\mathcal{P}(t)}\left|\partial_x\left[v(p,A_{n+1}(t,x))f_0(Z_n(0,t,x,p))\right]\right|dp\nonumber\\
  & \leq & C^0_{\bar{T}}\exp\left\{C^0_{\bar{T}}\int^t_0\left(1+\left\|\partial^2_x\Phi^n(\tau)\right\|_{L^{\infty}_x}+\left\|\partial^2_xA_n(\tau)\right\|_{L^{\infty}_x}\right)d\tau\right\}.\nonumber
\end{eqnarray}  
Hence, in view of Lemma \ref{Estimate Darwin Potentials}, we have for all $0\leq t\leq\bar{T}$ that
\begin{eqnarray}
\left\|\partial^2_x\Phi^{n+1}(t)\right\|_{L^{\infty}_x} & + & \left\|\partial^2_xA_{n+1}(t)\right\|_{L^{\infty}_x}\nonumber\\
& \leq & C^0_{\bar{T}}\left(1+\ln^+\left\|\partial_x\rho^{n+1}(t)\right\|_{L^{\infty}_x}+\ln^+\left\|\partial_xj_{n+1}(t)\right\|_{L^{\infty}_x}\right)\nonumber\\
& \leq & C^0_{\bar{T}}+C^0_{\bar{T}}\int^t_0\left(\left\|\partial^2_x\Phi^n(\tau)\right\|_{L^{\infty}_x}+\left\|\partial^2_xA_n(\tau)\right\|_{L^{\infty}_x}\right)d\tau.\nonumber
\end{eqnarray}  
Since the right-hand side is bounded for $n=0$, induction in $n$ yields
$$ \left\|\partial^2_x\Phi^n(t)\right\|_{L^{\infty}_x} + \left\|\partial^2_xA_n(t)\right\|_{L^{\infty}_x}\leq C^0_{\bar{T}}\exp\left\{C^0_{\bar{T}}\bar{T}\right\},
$$
for all $n\in\mathbb{N}$ and $0\leq t\leq\bar{T}$. In turn, this provides a uniform bound on the derivatives of the iterates for the current and density functions.\\
\\
\noindent\textbf{Step 4}: We show that $\left\{f^n\right\}$ is Cauchy in the uniform norm on $[0,\bar{T}]\times\mathbb{R}^6$. To start with, notice that 
\begin{eqnarray}
\label{Initial Estimate}
 \left|f^{n+1}(t,z)-f^n(t,z)\right| & = & \left|f_0(Z_n(0,t,z))-f_0(Z_{n-1}(0,t,z))\right|\nonumber\\
 & \leq & C\left|Z_n(0,t,z)-Z_{n-1}(0,t,z)\right|.
\end{eqnarray}
On the other hand, by using the estimates in the previous steps, it is not difficult to check that the characteristics equations lead to
\begin{eqnarray}
\left|X_n(s)\right. & - & \left. X_{n-1}(s)\right|\nonumber\\
   & \leq & \int^t_s\left|v(P_n(\tau),A_n(\tau,X_n(\tau)))-v(P_{n-1}(\tau),A_{n-1}(\tau,X_{n-1}(\tau)))\right|d\tau\nonumber\\
  & \leq & C\int^t_s\left(\left|X_n(\tau)-X_{n-1}(\tau)\right|+\left|P_n(\tau)-P_{n-1}(\tau)\right|\right.\nonumber\\
  & & \left.+\left\|A_n(\tau)-A_{n-1}(\tau)\right\|_{L^{\infty}_x}\right)d\tau,\nonumber
\end{eqnarray}
and
\begin{eqnarray}
\left|P_n(s)\right. & - & \left. P_{n-1}(s)\right| \nonumber\\
& \leq &    \int^t_s\left(\frac{}{}\left|\nabla\Phi^n(\tau,X_n(\tau))-\nabla\Phi^{n-1}(\tau,X_{n-1}(\tau))\right|\right.\nonumber\\
& & +\left|\frac{}{}\left(v^i_{A_n}\nabla A^i_n\right)\left(\tau,X_n(\tau),P_n(\tau)\right)\right.\nonumber\\
& & -\left.\left.\left(v^i_{A_{n-1}}\nabla A^i_{n-1}\right)\left(\tau,X_{n-1}(\tau),P_{n-1}(\tau)\right)\right|\right)d\tau\nonumber\\
& \leq &  C\int^t_s\left(\frac{}{}\left|X_n(\tau)-X_{n-1}(\tau)\right|+\left|P_n(\tau)-P_{n-1}(\tau)\right|\right.\nonumber\\
& & +\left\|\partial_x\Phi^n(\tau)-\partial_x\Phi^{n-1}(\tau)\right\|_{L^{\infty}_x}\nonumber\\
& & +\left.\left\|A_n(\tau)-A_{n-1}(\tau)\right\|_{L^{\infty}_x}+\left\|\partial_xA_n(\tau)-\partial_xA_{n-1}(\tau)\right\|_{L^{\infty}_x}\right) d\tau\nonumber.
\end{eqnarray}
Therefore, after adding the above expressions, Gronwall's inequality yields
\begin{eqnarray}
\label{Full Characteristics Estimate}
  \left|Z_n(0,t,z)\right. & - & \left. Z_{n-1}(0,t,z)\right| \leq  C\int^t_0\left(\left\|\partial_x\Phi^n(\tau)-\partial_x\Phi^{n-1}(\tau)\right\|_{L^{\infty}_x}\right.\nonumber\\
  & + &\left. \left\|A_n(\tau)-A_{n-1}(\tau)\right\|_{L^{\infty}_x}+\left\|\partial_xA_n(\tau)-\partial_xA_{n-1}(\tau)\right\|_{L^{\infty}_x}\right) d\tau.\hspace{.5cm}
\end{eqnarray}

Now, to produce a Gronwall's inequality resulting from (\ref{Initial Estimate}) and (\ref{Full Characteristics Estimate}), we look for suitable estimates on the right-hand side of (\ref{Full Characteristics Estimate}). To start with, let $R=\max\left\{\bar{X}_0+\bar{T},\mathcal P(\bar{T})\right\}$. For all $n\in\mathbb{N}$ and $0\leq t\leq\bar{T}$ we have 
$$ \texttt{supp}f^n(t)\subset B_R\times B_R.
$$ 
Linearity and Lemma \ref{Estimate Darwin Potentials} yield
\begin{eqnarray}
\label{Last Estimate Scalar Potential}
 \left\|\partial_x\Phi^n(\tau)-\partial_x\Phi^{n-1}(\tau)\right\|_{L^{\infty}_x} & \leq & C \left\|\rho^n(\tau)-\rho^{n-1}(\tau)\right\|^{1/3}_{L^1_x}\left\|\rho^n(\tau)-\rho^{n-1}(\tau)\right\|^{2/3}_{L^{\infty}_x}\nonumber\\
  & \leq & C_R\left\|f^n(\tau)-f^{n-1}(\tau)\right\|_{L^{\infty}_{x,p}}.
\end{eqnarray}
To estimate the terms involving the vector potential, we proceed as follows. 

According to the definition of the iterates, it is clear that for each $n\in \mathbb{N}$ they satisfy
$$ \partial_tf^{n+1}+v_{A_n}\cdot\nabla_xf^{n+1}-\left[\nabla\Phi^n-v^i_{A_n}\nabla A^i_n\right]\cdot\nabla_pf^{n+1}=0.
$$
Hence, by the uniqueness proof now in terms of the iterates, (see (\ref{Estimate For Gronwall 2})), 
\begin{eqnarray}
\label{AL2fLinfty}
\left\|A_n(\tau)-A_{n-1}(\tau)\right\|_{L^2_x(B_R)} & \leq & C_R \left\|f^{n+1}(\tau)-f^n(\tau)\right\|_{L^2_{x,p}}\nonumber\\
& \leq & C_R\left\|f^{n+1}(\tau)-f^n(\tau)\right\|_{L^{\infty}_{x,p}}.
\end{eqnarray}
Next, we claim that 
\begin{eqnarray}\label{eq difference vector potentials}
\left\|A_n(\tau)\right. & - & \left. A_{n-1}(\tau)\right\|_{L^{\infty}_x}\nonumber\\   
& \leq & C_R\left(\left\|f^n(\tau)-f^{n-1}(\tau)\right\|_{L^{\infty}_{x,p}} 
+\left\|A_n(\tau)-A_{n-1}(\tau)\right\|_{L^2_x(B_R)}\right).\hspace{.5cm}
\end{eqnarray}

%%%%%%%%%%%%%%%%%%%%%%%

Indeed, by (\ref{Integral An}), we have
\begin{eqnarray}\label{eq difference vector potentials 1}
A_n(\tau,x)-A_{n-1}(\tau,x) &=& \frac{1}{2}\int_{\R^3}\int_{\R^3} \K(x,y)\left[v_{A_n}f^n-v_{A_{n-1}}f^{n-1}\right](\tau,y,p)dy dp \nonumber\\
&=& \frac{1}{2}\int_{\R^3}\int_{\R^3} \K(x,y)\left[v_{A_n}-v_{A_{n-1}}\right] f^n(\tau,y,p) dy dp \nonumber\\
& & + \frac{1}{2}\int_{\R^3}\int_{\R^3} \K(x,y)v_{A_{n-1}}\left[f^n - f^{n-1}\right](\tau,y,p)dydp \nonumber\\
&=& I_1(\tau, x) + I_2(\tau, x)
\end{eqnarray}
where $\mathcal{K}(x,y)=\left|y-x\right|^{-1}\left[\texttt{id}+\omega\otimes\omega\right]$. Since $|v_{A_{n}}| \leq 1$, $|\K(x,y)| \leq C|y-x|^{-1}$ and $\texttt{supp}f^n(t)\subset B_R\times B_R$ for all $n\in\N$, we have that 
\begin{eqnarray*}
|I_2(\tau, x)|  & =&  \frac{1}{2} \Big{|}\int_{\R^3}\int_{\R^3} \K(x,y)v_{A_{n-1}}\left[f^n - f^{n-1}\right](\tau,y,p)dydp \Big{|} \\
&\leq& C \int_{B_R} \eta(\tau, y)\frac{dy}{|y-x|}
\end{eqnarray*}
where we defined $\eta(\tau, y) = \int_{B_R} |f^n - f^{n-1}|(\tau, y, p) dp$ and $\texttt{supp}\,\eta(\tau)\subset B_R$. Note that $\|\eta(\tau)\|_{L^\infty(\R^3)} \leq  C_R \| f^n(\tau) - f^{n-1}(\tau)\|_{L^\infty_{x,p}}$. Then using Lemma \ref{Pallard}, we have,
\begin{eqnarray}\label{eq difference vector potentials 2}
|I_2(\tau, x)|  & \leq & C \Big{\|} \int_{B_R} \eta(\tau, y)\frac{dy}{|y-x|}\Big{\|}_{L^\infty_x} \nonumber \\
&\leq& C \|\eta(\tau)\|^{2/3}_{L^1_x} \|\eta(\tau)\|^{1/3}_{L^\infty_x}  \leq C_R \|\eta(\tau)\|_{L^\infty_x}\nonumber\\
& \leq&  C_R \| f^n(\tau) - f^{n-1}(\tau)\|_{L^\infty_{x,p}}.
\end{eqnarray}
On the other hand, we also have
\begin{eqnarray*}
|I_1(\tau,x)| &=& \frac{1}{2} \Big{|} \int_{\R^3}\int_{\R^3} \K(x,y)\left[v_{A_n}-v_{A_{n-1}}\right] f^n(\tau,y,p) dy dp \Big{|}\\
& \leq & C  \int_{B_R}\int_{B_R} |v_{A_n}-v_{A_{n-1}}| f^n(\tau,y,p) \frac{dy dp}{|y-x|}.
\end{eqnarray*}
Then, since $v_A=v(g_A)$ where $g_A(\tau,y,p) = p-A(\tau,y)$ and 
$$v(z)=\frac{z}{\sqrt{1+|z|^2}} \in C^1_b(\R^3;\R^3),$$
we have by the mean value theorem that for all $(y,p)\in B_R\times B_R$,
\[ \big{|} \left[ v_{A_n} - v_{A_{n-1}}\right] (\tau,y,p) \big{|} \leq C_R |A_n -A_{n-1}|(\tau, y).\]
Therefore,
\[|I_1(\tau,x)| \leq C_R \Big{\|} \int_{B_R} \rho^n (\tau, y) |A_n -A_{n-1}|(\tau, y)\frac{dy}{ |y-x|} \Big{\|} _{L^\infty_x}.\]
We use Lemma \ref{Pallard} and Cauchy-Schwarz inequality to have that 
\begin{eqnarray*}
|I_1(\tau,x)| &\leq& C_R \| \rho^n(\tau)|A_n -A_{n-1}|(\tau)\|^{1/3}_{L^1_x(B_R)} \| \rho^n(\tau)|A_n -A_{n-1}|(\tau)\|^{2/3}_{L^2_x(B_R)} \\
&\leq & C_R \|\rho^n(\tau)\|^{1/3}_{L^2_x(B_R)}  \|\rho^n(\tau)\|^{2/3}_{L^\infty_x} \|(A_n -A_{n-1})(\tau)\|_{L^2_x(B_R)}\\
&\leq& C_R  \|\rho^n(\tau)\|_{L^\infty_x} \|A_n(\tau)-A_{n-1}(\tau)\|_{L^2_x(B_R)}.
\end{eqnarray*}
Hence (\ref{eqn estimate rho n}) and Step 2 imply that
\begin{equation}\label{eq difference vector potentials 3}
|I_1(\tau,x)| \leq C_R  \|A_n(\tau)-A_{n-1}(\tau)\|_{L^2_x(B_R)}.
\end{equation}
By combining (\ref{eq difference vector potentials 1}) - (\ref{eq difference vector potentials 3}), the estimate (\ref{eq difference vector potentials}) readily follows.

%%%%%%%%%%%%%%%%%%%%%%%%%%%%%

Also, if we write respectively $f^n-f^{n-1}$ and $A_n-A_{n-1}$ in Lemma \ref{LInfinity Estimate Double Time Derivative A} instead of $\partial_tf$ and $\partial_tA$, we find that
\begin{eqnarray}
\label{eq difference derivative vector potentials}
\left\|\partial_xA_n(\tau) \right. & - & \left. \partial_xA_{n-1}(\tau)\right\|_{L^{\infty}_x}\nonumber\\ 
& \leq &   C_R\left(\left\|f^n(\tau)-f^{n-1}(\tau)\right\|_{L^{\infty}_{x,p}}+\left\|\left(A_n(\tau)-A_{n-1}(\tau)\right)\right\|_{L^{\infty}_x}\right).\hspace{.5cm}
\end{eqnarray}
Therefore, the estimates (\ref{AL2fLinfty}), (\ref{eq difference vector potentials}) and (\ref{eq difference derivative vector potentials}) yield 
\begin{eqnarray}
\label{Last Estimate Vector Potential}
\left\|A_n(\tau) \right. & - & \left. A_{n-1}(\tau)\right\|_{L^{\infty}_x} + \left\|\partial_xA_n(\tau)-\partial_xA_{n-1}(\tau)\right\|_{L^{\infty}_x}\nonumber\\
& \leq & C_R\left(\left\|f^{n+1}(t)-f^n\right\|_{L^{\infty}_{x,p}}+ \left\|f^n(\tau)-f^{n-1}(\tau)\right\|_{L^{\infty}_{x,p}}\right).
\end{eqnarray}
Hence, if we combine (\ref{Initial Estimate}) and (\ref{Full Characteristics Estimate}) with (\ref{Last Estimate Scalar Potential}) and (\ref{Last Estimate Vector Potential}), a use of Gronwall's lemma gives
$$ \left\|f^{n+1}(t)-f^n(t)\right\|_{L^{\infty}_{x,p}} \leq C_R\int^t_0\left\|f^n(\tau)-f^{n-1}(\tau)\right\|_{L^{\infty}_{x,p}}d\tau,
$$
which by induction, readily implies the claim. It follows that $\left\{f^n\right\}$ converges uniformly to some $f\in C([0,\bar{T}]\times\mathbb{R}^6;\mathbb{R})$ and for all $0\leq t\leq\bar{T}$ we have
$$ \texttt{supp}f(t)\subset B_R\times B_R.
$$
Finally, if we respectively define $\rho$, $\Phi$ and $A$ according to (\ref{Density and Current}), (\ref{Phi In RVD}) and (\ref{A In RVD}), we have that $\rho$, $\Phi$ and $A$ are $C_b$, and $\rho^n\rightarrow\rho$, $\Phi^n\rightarrow\Phi$ and $A_n\rightarrow A$ hold uniformly on $[0,\bar{T}]\times\mathbb{R}^3$. The latter follows from (\ref{Last Estimate Vector Potential}). The uniform limits $v_{A_n}\rightarrow v_A$ and $v_{A_n}f^n\rightarrow v_Af$ can be easily checked, and therefore $j_n\rightarrow j_A$ uniformly on $[0,\bar{T}]\times\mathbb{R}^3$, with $j_A\in C_b([0,\bar{T}]\times\mathbb{R}^3;\mathbb{R}^3)$ defined by (\ref{Density and Current}). \\
\\
\noindent\textbf{Step 5}: Actually $f\in C^1(I\times\mathbb{R}^6;\mathbb{R})$, as we show next. Indeed, in view of Step 4 and, respectively, (\ref{Last Estimate Scalar Potential}) and (\ref{Last Estimate Vector Potential}), the sequences $\left\{\partial_x\Phi^n\right\}$ and $\left\{\partial_xA_n\right\}$ are uniformly Cauchy on $[0,\bar{T}]\times\mathbb{R}^3$. Moreover, by Lemma \ref{Estimate Darwin Potentials} we have
\begin{eqnarray}
\left\|\partial^2_xA_m(t) \right. & - & \left.\partial^2_xA_n(t)\right\| \leq C\left[R^{-3}\left\|j_m(t)-j_n(t)\right\|_{L^1_x}\right.\nonumber\\
& + & \left.h\left\|\partial_xj_m(t)-\partial_xj_n(t)\right\|_{L^{\infty}_x}+\left(1+\ln\left(R/h\right)\left\|j_m(t)-j_n(t)\right\|_{L^{\infty}_x}\right)\right]\nonumber
\end{eqnarray}
and similarly for $\partial^2_x\Phi^n(t)$. Hence, the known estimates and the fact that we can choose $h$ arbitrary small imply that $\left\{\partial^2_x\Phi^n\right\}$ and $\left\{\partial^2_xA_n\right\}$ are uniformly Cauchy on $[0,\bar{T}]\times\mathbb{R}^3$ as well. Therefore, we have (with a slight abuse of notation)
$$  (\Phi,A),\;\partial_x(\Phi,A),\; \partial^2_x(\Phi,A)\; \in C([0,\bar{T}]\times\mathbb{R}^3;\mathbb{R}\times\mathbb{R}^3),
$$
and so the characteristic flow $Z\in C^1([0,\bar{T}]\times[0,\bar{T}]\times\mathbb{R}^6)$ induced by the limiting field is in turn the limit of the sequence $\left\{Z_n\right\}$. As a result, the function
$$ f(t,z)=\lim_{n\rightarrow\infty}f_0(Z_n(0,t,z))=f_0(Z(0,t,z))
$$
has the claimed regularity and the triplet $(f,\Phi,A)$ satisfies (\ref{Vlasov In RVD})-(\ref{Velocity In RVD}).\\
\\
\noindent\textbf{Step 6}: We show that the potentials $\Phi$ and $A$ have the required regularity in time. Define $f_{\lambda}=\lambda f^m+\left(1-\lambda\right)f^n$, $0\leq\lambda\leq1$. This definition is just like the one in the uniqueness proof but in terms of any two elements of the sequence $\left\{f^n\right\}$. Let $A_{\lambda}$ be the vector potential induced by $f_{\lambda}$. Following the lines in the proof of Lemma \ref{L2 Estimate Time Derivative A}, it is not difficult to check that (this is analogous to (\ref{Working Integrals}))
\begin{eqnarray}
\label{Long And Painful}
\int_{\mathbb{R}^3}\left|\partial_{\lambda}\partial_t\partial_xA_{\lambda}\right|^2 & = & \int_{B_R}\int_{B_R}f_{\lambda}\partial_{\lambda}\partial_tA_{\lambda}\cdot\partial_{\lambda}\partial_tv_{A_{\lambda}}\nonumber\\
&  + & \int_{B_R}\int_{B_R}\partial_{\lambda}\partial_tA\cdot\left[\partial_{\lambda}v_{A_\lambda}\partial_tf_{\lambda}+\partial_tv_{A_\lambda}\partial_{\lambda}f_{\lambda}+v_{A_\lambda}\partial_{\lambda}\partial_tf_{\lambda}\right]\nonumber\\
& + & \partial_{\lambda}\partial_{t}\int_{\mathbb{R}^3}\int_{B_R}\frac{1}{r}\left(\nabla\cdot A_{\lambda}\right)\left(\nabla\cdot\partial_{\lambda}\partial_tj_{\lambda}\right).
\end{eqnarray}
Since $\nabla\cdot A_{\lambda}=0$, the third integral in the right-hand side vanishes. On the other hand, by using the notation of the Appendix, we have $\partial_{\lambda}v_{A_{\lambda}}=-Dv_{A_{\lambda}}\partial_{\lambda}A_{\lambda}$, also $\partial_tv_{A_{\lambda}}=-Dv_{A_{\lambda}}\partial_tA_{\lambda}$, and
$$ \partial_{\lambda}\partial_tv_{A_{\lambda}}=-Dv_{A_{\lambda}}\partial_{\lambda}\partial_tA-D^2v_{A_{\lambda}}\partial_{\lambda}A_{\lambda}\partial_tA_{\lambda}.
$$

Therefore, since by Step 5 $\left|\partial_tf_{\lambda}\right|\leq C_R$, and by Corollary \ref{Corollary Estimate Time Derivative} $\left|\partial_tA_{\lambda}\right|\leq C_R$, we obtain from (\ref{Long And Painful}) that
\begin{eqnarray} 
\label{Estimate With Gamma}
\lefteqn{\int_{\mathbb{R}^3}\left|\partial_{\lambda}\partial_t\partial_xA_{\lambda}\right|^2 +\int_{B_R}\int_{B_R}\frac{f_{\lambda}}{\sqrt{1+g^2_{\lambda}}}\left(\left|\partial_{\lambda}\partial_tA_{\lambda}\right|^2-\left|v_{A_{\lambda}}\cdot\partial_{\lambda}\partial_tA_{\lambda}\right|^2\right)}\nonumber\\
& \leq & C_R\int_{B_R}\int_{B_R}\left|\partial_{\lambda}\partial_tA_{\lambda}\right|\left[\left|\partial_{\lambda}A_\lambda\right|+\left|\partial_{\lambda}f_{\lambda}\right|+\left|\partial_{\lambda}\partial_tf_{\lambda}\right|\right]\nonumber\\
& \leq & C_R\left\|\partial_{\lambda}\partial_tA_{\lambda}(t)\right\|_{L^2_x}\left[\left\|\partial_{\lambda}A_{\lambda}(t)\right\|_{L^2_x(B_R)}+\left\|\partial_\lambda f_{\lambda}(t)\right\|_{L^{\infty}_{x,p}}+\left\|\partial_{\lambda}\partial_tf_{\lambda}(t)\right\|_{L^{\infty}_{x,p}}\right].\hspace{-.7cm}
\end{eqnarray}
In the last step we have used the Cauchy-Schwarz inequality. Now, define 
$$ G_{mn} = \sup_{0\leq t\leq\bar{T}}\left( \left\|f^n(t)-f^m(t)\right\|_{L^{\infty}_{x,p}}+\left\|\partial_tf^n(t)-\partial_tf^m(t)\right\|_{L^{\infty}_{x,p}}\right),
$$
which in view of Steps 4 and 5 converges to zero as $n,m\rightarrow\infty$. If we use the estimate (\ref{Estimate Needed 2}) for the iterates, i.e. $\left\|\partial_{\lambda}A_{\lambda}(t)\right\|_{L^2_x(B_R)}\leq C_R\left\|\partial_{\lambda}f_{\lambda}(t)\right\|_{L^2_{x,p}}$, we find that the expression in square brackets in the right-hand side of (\ref{Estimate With Gamma}) can be estimated as
\begin{eqnarray*}
\lefteqn{\left\|\partial_{\lambda}A_{\lambda}(t)\right\|_{L^2_x(B_R)}+\left\|\partial_\lambda f_{\lambda}(t)\right\|_{L^{\infty}_{x,p}}+\left\|\partial_{\lambda}\partial_tf_{\lambda}(t)\right\|_{L^{\infty}_{x,p}}}\\
& \leq & C_R\left(\left\|\partial_{\lambda}f_{\lambda}(t)\right\|_{L^2_{x,p}}+\left\|\partial_\lambda f_{\lambda}(t)\right\|_{L^{\infty}_{x,p}}+\left\|\partial_{\lambda}\partial_tf_{\lambda}(t)\right\|_{L^{\infty}_{x,p}}\right)\ \leq\ C_RG_{mn},
\end{eqnarray*}
uniformly in $\lambda$. On the other hand, since $\left|v_{A_{\lambda}}\right|<1$ strictly, we can reason as in the proof of Lemma \ref{L2 Estimate Time Derivative A} to find a lower bound on the left-hand side of (\ref{Estimate With Gamma}). This lower bound can then be estimated as 
$$
\left\|\partial_{\lambda}\partial_t\partial_xA_{\lambda}(t)\right\|^2_{L^2_x} +\left\|\rho^{1/2}_{\lambda}(t)\partial_{\lambda}\partial_tA_{\lambda}(t)\right\|^2_{L^2_x(B_R)}\leq C_R\left\|\partial_{\lambda}\partial_tA_{\lambda}(t)\right\|_{L^2_x(B_R)}G_{mn}.
$$
Consider the first term on the left-hand side. Poincar{\'e}'s inequality and the above estimate imply that
$$
\left\|\partial_{\lambda}\partial_tA_{\lambda}(t)\right\|^2_{L^2_x(B_R)}\leq C_R\left\|\partial_{\lambda}\partial_t\partial_xA_{\lambda}(t)\right\|^2_{L^2_x}\leq C_R\left\|\partial_{\lambda}\partial_tA_{\lambda}(t)\right\|_{L^2_x(B_R)}G_{mn}.
$$
Then, the last two estimates yield
\begin{equation}
\label{Estimate L^2 Gamma} 
\left\|\partial_{\lambda}\partial_tA_{\lambda}(t)\right\|_{L^2_x(B_R)}+\left\|\partial_{\lambda}\partial_t\partial_xA_{\lambda}(t)\right\|_{L^2_x}\leq C_RG_{mn}.
\end{equation}
On the other hand, by the definition of $A_{\lambda}$, we have 
$$ \partial_{\lambda}\partial_tA_{\lambda}(t,x)=\int_{B_R}\int_{B_R}\mathcal{K}(x,y)\partial_{\lambda}\partial_t\left[v_{A_{\lambda}}f_{\lambda}(t,x,p)\right]dpdy.
$$
Therefore, after taking the product rule in the integrand, we may proceed as in Lemma \ref{LInfinity Estimate Time Derivative A} to obtain the estimate
\begin{eqnarray}
\label{Other LInfinity Estimate Time Derivative A}
\left\|\partial_{\lambda}\partial_tA_{\lambda}(t)\right\|_{L^{\infty}_x}\leq C_R\left(G_{mn}+\left\|\partial_{\lambda}\partial_tA_{\lambda}(t)\right\|_{L^2_{x}(B_R)}\right),
\end{eqnarray}
where again we have used $\left|\partial_tf_{\lambda}\right|\leq C_R$ and $\left|\partial_tA_{\lambda}\right|\leq C_R$. Similarly, since 
$$ \partial_{\lambda}\partial_t\partial_xA_{\lambda}(t,x)=\int_{B_R}\int_{B_R}\partial_x\mathcal{K}(x,y)\partial_{\lambda}\partial_t\left[v_{A_{\lambda}}f_{\lambda}(t,x,p)\right]dpdy,
$$
we can proceed as in Lemma \ref{LInfinity Estimate Double Time Derivative A} to find
\begin{eqnarray}
\label{Other LInfinity Estimate Double Time Derivative A} 
\left\|\partial_{\lambda}\partial_t\partial_xA_{\lambda}(t)\right\|_{L^{\infty}_x}\leq C_R\left(G_{mn}+\left\|\partial_{\lambda}\partial_tA_{\lambda}(t)\right\|_{L^{\infty}_{x}(B_R)}\right).
\end{eqnarray}
Hence, since $\left|\partial_tA_m-\partial_tA_n\right|\leq\int^1_0\left|\partial_{\lambda}\partial_tA_{\lambda}\right|d\lambda\leq\sup_{\lambda}\left|\partial_{\lambda}\partial_tA_{\lambda}\right|$ and similarly for $\left|\partial_t\partial_xA_m-\partial_t\partial_xA_n\right|$, we can gather the above estimates to find that 
$$ \left\|\partial_tA_m(t)-\partial_tA_n(t)\right\|_{L^{\infty}_x}+\left\|\partial_t\partial_xA_m(t)-\partial_t\partial_xA_n(t)\right\|_{L^{\infty}_x}\leq C_RG_{mn}.
$$
Therefore, the sequences $\left\{\partial_tA_n\right\}$ and $\left\{\partial_t\partial_xA_n\right\}$ are uniformly Cauchy and we have that $\partial_tA_n\rightarrow \partial_tA$ and $\partial_t\partial_xA_n\rightarrow\partial_t\partial_xA$ uniformly on $[0,\bar{T}]\times\mathbb{R}^3$. In turn, the former limit and Steps 4 and 5 imply the uniform convergence $\partial_{t}(v_{A_n}f^n)\rightarrow\partial_{t}(v_Af)$, and so $\partial_tj_n\rightarrow\partial_tj_{A}$. Also, $\partial_t\rho^n\rightarrow\partial_t\rho$. Hence, just as in Step 5, the sequences $\left\{\partial_t\partial^2_x\Phi\right\}$ and $\left\{\partial_t\partial^2_xA\right\}$ are uniformly Cauchy on $[0,\bar{T}]\times\mathbb{R}^3$. Therefore, since trivially $\partial_t\Phi$ and $\partial_t\partial_x\Phi$ are continuous on $[0,\bar{T}]\times\mathbb{R}^3$, we conclude that 
 $$  \partial_t(\Phi,A),\; \partial_t\partial_x(\Phi,A),\; \partial_t\partial^2_x(\Phi,A) \in C([0,\bar{T}]\times\mathbb{R}^3).
 $$
Having proved the claim, and since $0\leq\bar{T}<T$ was arbitrary, we conclude that $f\in C^1([0,T[\times\mathbb{R}^6;\mathbb{R})$ is a classical solution of the relativistic Vlasov-Darwin system.
$$ \vspace{-.3cm}
$$
\noindent\textbf{Step 7}: Moreover, $f(t)\in C^{1,\alpha}(\mathbb{R}^6;\mathbb{R})$, $0<\alpha<1$, for each $0\leq t<T$. In view of Remark \ref{Alpha In Vlasov}, this holds if $\left(\Phi,A\right)(t)\in C^{2,\alpha}(\mathbb{R}^3;\mathbb{R}\times\mathbb{R}^3)$. But, since we have $(\rho,j_A)(t)\in C^1_0(\mathbb{R}^3;\mathbb{R}\times\mathbb{R}^3)\subset C^{\alpha}_0(\mathbb{R}^3;\mathbb{R}\times\mathbb{R}^3)$, the regularity needed for the potentials is guaranteed (see the last lines in the proof of Lemma \ref{Fixed Point Lemma}).\\
\\
\noindent\textbf{Step 8}: The proof of the continuation criterion is as follows. Let $f$ be the solution of the RVD system previously obtained, which clearly satisfies $\left.f\right|_{t=0}$. As shown in (\ref{Life Span}), the life span of $f$ is $T\equiv\left(C(f_0)\bar{P}_0\right)^{-1}$ with
$$ C(f_0)=3(2\pi)^{2/3}\left\|f_0\right\|^{1/3}_{L^1_{x,p}}\left\|f_0\right\|^{2/3}_{L^{\infty}_{x,p}}.
$$
Define $\bar{P}_T=\sup\left\{\left|p\right|:\exists0\leq t<T,x\in\mathbb{R}^3:f(t,x,p)\neq0\right\}$ and assume that $\bar{P}_T<\infty$ but $T<\infty$. We claim that this is a contradiction.

Fix $0<t_0<T$ and consider $f(t_0)$ as a Cauchy datum of the RVD system, which is guaranteed by Step 7. Known estimates yield
$$  \left\|f(t_0)\right\|_{L^1_{x,p}}=\left\|f_0\right\|_{L^1_{x,p}},\quad \left\|f(t_0)\right\|_{L^{\infty}_{x,p}}=\left\|f_0\right\|_{L^{\infty}_{x,p}}.
$$
Thus, $C(f(t_0))=C(f_0)$. Define $\epsilon=\left(C(f_0)\bar{P}_T\right)^{-1}$, which does \textit{not} depend on $t_0$. Steps 1-3 imply that all uniform estimates on the sequence of approximate solutions induced by $f(t_0)$ hold on $[t_0,t_0+\epsilon[$. Then, $f(t_0)$ yields a unique classical solution of the RVD system on that interval. 

But we could have fixed $t_0$ arbitrary close to the life span $T<\infty$ of $f$ and so extend this solution beyond $T$, which is a contradiction. Hence, we have shown that $\bar{P}_T<\infty$ implies $T=\infty$. This, and the uniqueness result, conclude the proof of Theorem \ref{Local Solutions Theorem}.\qed

%%%%%%%%%%%%%%%%%%%%%%%%%%%%%%%%%%%%%%%%%%%%%%%%%%%%%%%%%%
%%%%%%%%%%%%%%%%%%%%%%%%%%%%%%%%%%%%%%%%%%%%%%%%%%%%%%%%%%

\subsection{Global Solutions}
\label{Global Solutions Section}

If additional conditions are imposed on the Cauchy datum in Theorem \ref{Local Solutions Theorem}, then the local solution found in the previous section can be extended globally in time. We prove this result next. We start by defining the set where the Cauchy datum will be taken from. For $\bar{X}_0>0, \,\bar{P}_0>0$ and $0<\alpha<1$ given, let 
\begin{equation}\label{definition of D}
 \mathcal{D}  = \left\{f\in C^{1,\alpha}(\mathbb{R}^6;\mathbb{R}): \; f\geq0,\:\left\|f\right\|_{W^{1,\infty}_{x,p}}\leq1,\:\texttt{supp}f\subset B_{\bar{X}_0}\times B_{\bar{P}_0}\right\}.
\end{equation}

\begin{theorem}
\label{Global Solutions Theorem}
There exists a $\delta>0$ such that, if $f_0\in\mathcal{D}$ with $\left\|f_0\right\|_{L^{\infty}_{x,p}}\leq\delta$, then the classical solution of the RVD system (\ref{Vlasov In RVD})-(\ref{Velocity In RVD}) with Cauchy datum $f_0$ is global in time. Moreover, for $t>0$ this solution satisfies the decay estimates
\begin{eqnarray}
  \left\|\rho(t)\right\|_{L^{\infty}_x}+\left\|j_A(t)\right\|_{L^{\infty}_x} & \leq & Ct^{-3}\\
  \left\|\partial_x\Phi(t)\right\|_{L^{\infty}_x}+\left\|\partial_xA(t)\right\|_{L^{\infty}_x} +  \left\|\partial_tA(t)\right\|_{L^{\infty}_x}  & \leq & Ct^{-2}\\
  \left\|\partial^2_x\Phi(t)\right\|_{L^{\infty}_x}+\left\|\partial^2_xA(t)\right\|_{L^{\infty}_x} & \leq & Ct^{-3}\ln(1+t).
\end{eqnarray}  
\end{theorem}

\vspace{.5cm}

We first introduce some technical results and postpone the actual proof of Theorem \ref{Global Solutions Theorem} to the end of this section. The following lemma shows that a sufficiently small Cauchy datum leads to a classical solution of the RVD system which exists on any given time interval and induces potentials whose derivatives can be made as small as desired. 

\begin{lemma}\label{lemma_noname}
 Fix $\epsilon>0$ and $T>0$. There exists $\delta=\delta(\epsilon,T)>0$ such that, if $f_0\in\mathcal{D}$ with $\left\|f_0\right\|_{L^{\infty}_{x,p}}\leq\delta$, then the classical solution of the RVD system with Cauchy datum $f_0$ exists on the time interval $\left[0,T\right]$ and induces potentials satisfying 
 \begin{equation}
 \label{Epsilon Bound}
\left\|\partial_tA(t)\right\|_{L^{\infty}_x}+\left\|\partial_xA(t)\right\|_{L^{\infty}_x}+\left\|\partial_x\Phi(t)\right\|_{L^{\infty}_x}+\left\|\partial^2_xA(t)\right\|_{L^{\infty}_x}+\left\|\partial^2_x\Phi(t)\right\|_{L^{\infty}_x}<\epsilon
 \end{equation}
 for all $0\leq t\leq T$.
\end{lemma}

\begin{proof} In view of Lemma \ref{Estimate Darwin Potentials}, and since $\left|j_A\right|\leq \rho$ and
$$\left\|\partial_x\rho(t)\right\|_{L^{\infty}_{x}}+\left\|\partial_xj_A(t)\right\|_{L^{\infty}_{x}}\leq C^0_{T},\quad 0\leq t\leq T
$$
hold (the latter proved just as in Step 3 in Theorem \ref{Local Solutions Theorem}, with the estimates applied to the solution instead of the iterates), the \textit{space} derivatives of $A$ satisfy the same estimates as the space derivatives of $\Phi$. Hence, the proof is mutatis mutandis the proof of \cite[Lemma 4.2]{Rein} for the Vlasov-Poisson system, as far as the space derivatives of the potentials are concerned. As for $\partial_tA$, the result follows suit in view of the estimates in Lemmas \ref{L2 Estimate Time Derivative A} and \ref{LInfinity Estimate Time Derivative A}.
\end{proof}

To proceed, we now define the so-called \textit{free streaming condition} for classical solutions of the RVD system. 

\begin{definition}
\label{FS Alpha}
  Fix $\beta>0$ and $a>0$. A classical solution of the RVD system is said to satisfy the free streaming condition of parameter $\beta$ (FS$\beta$) on the time interval $[0,a]$, if it exists on $[0,a]$ and induces potentials satisfying the estimates 
 \begin{eqnarray}
\left\|\partial_tA(t)\right\|_{L^{\infty}_x}+\left\|\partial_xA(t)\right\|_{L^{\infty}_x}+\left\|\partial_x\Phi(t)\right\|_{L^{\infty}_x} & \leq & \beta(1+t)^{-3/2},\nonumber\\
 \left\|\partial^2_xA(t)\right\|_{L^{\infty}_x}+\left\|\partial^2_x\Phi(t)\right\|_{L^{\infty}_x} & \leq & \beta(1+t)^{-5/2},\nonumber
  \end{eqnarray}
for all $0\leq t\leq a$.  
\end{definition}

\begin{lemma}
\label{Lemma Estimates}
 There exists $\delta>0$, $\beta>0$ and a positive $C=C(\bar{X}_0,\bar{P}_0)$ such that any classical solution $f$ of the RVD system having a Cauchy datum $f_0\in\mathcal{D}$ with $\left\|f_0\right\|_{L^{\infty}_{x,p}}\leq\delta$ and satisfying (FS$\beta$) on some interval $[0,a]$, also satisfies the estimates
 \begin{eqnarray}
   \label{First Density Wanted}
  \left\|\rho(t)\right\|_{L^{\infty}_x}+\left\|j_A(t)\right\|_{L^{\infty}_x} & \leq & Ct^{-3},\\
   \label{Second Density Wanted}
  \left\|\partial_x\rho(t)\right\|_{L^{\infty}_x}+\left\|\partial_xj_A(t)\right\|_{L^{\infty}_x} & \leq & C.
   \end{eqnarray}
  for all $0\leq t\leq a$.
\end{lemma}

\begin{proof}
To prove (\ref{First Density Wanted})-(\ref{Second Density Wanted}), we first introduce some technical results which we present as a sequence of steps.\\
\\
\textbf{Step 1}: Let $0\leq s\leq t\leq a$. Denote by $(X,P)(s)=(X,P)(s,t,x,p)$ the solution of the characteristic system 
\begin{eqnarray}
   \dot{X}(s) & = & v_A(s,X(s),P(s))  \label{charact X} \\
   \dot{P}(s) & = & -\left[\nabla\Phi+v_A^i\nabla A^i\right](s,X(s),P(s)), \label{charact P}
\end{eqnarray}
with $(X,P)(t)=(x,p)$. Denote also $Dv_A(s)=Dv_A(P(s),A(s,X(s)))$, where the matrix $Dv_A$ is as given in the Appendix. Consider the system
\begin{eqnarray}
  \xi(s) & = & \partial_pX(s)-(s-t)Dv_A(t)\nonumber\\
  \eta(s) & = & Dv_A(s)\partial_pP(s)-Dv_A(t).\nonumber
\end{eqnarray}
Notice that $\xi(t)=\eta(t)=0$. We show that for some $C=C(\bar{X}_0,\bar{P}_0)>0$
\begin{equation} 
\label{First Claim}
\left|\xi(s)\right|\leq \beta Ce^{\beta C}(t-s).
\end{equation}  
Indeed, on the characteristic curves, we have 
\begin{eqnarray}
   \dot{\xi}(s) & = & \partial_p\dot{X}(s)-Dv_A(t)\nonumber\\
   & = & Dv_A(s)\left[\partial_pP(s)-\partial_xA(s,X(s))\partial_pX(s)\right]-Dv_A(t)\nonumber\\
   & = & \eta(s)-Dv_A(s)\partial_xA(s,X(s))\left[\xi(s)+(s-t)Dv_A(t)\right].\nonumber
\end{eqnarray}
Therefore, since $\left|Dv_A(s)\right|\leq C$, a use of (FS$\beta$) yields
\begin{eqnarray}
\label{Xi}
   \left|\xi(s)\right| & \leq & \int^t_s\left|\eta(\tau)\right|d\tau+\beta C\int^t_s\left(1+\tau\right)^{-3/2}\left[\left|\xi(\tau)\right|+(t-\tau)\right]d\tau\nonumber\\
   & \leq & \beta Ce^{\beta C}\left(\left(t-s\right)+\int^t_s\left|\eta(\tau)\right|d\tau\right),
\end{eqnarray}
where the Gronwall's inequality has been used in the last step. 

On the other hand, we have
\begin{eqnarray}
   \dot{\eta}(s) & = & Dv_A(s)\partial_p\dot{P}(s)+D^2v_A(s)\left[\dot{P}(s)-\dot{A}(s,X(s))\right]\partial_pP(s).\nonumber
\end{eqnarray}
In view of the characteristic system, it is not difficult to check that 
\begin{eqnarray}
 \left|\partial_p\dot{P}(s)\right| & \leq & C\left(\left\|\partial^2_x\Phi(s)\right\|_{L^{\infty}_x}+\left\|\partial^2_xA(s)\right\|_{L^{\infty}_x}+\left\|\partial_xA(s)\right\|^2_{L^{\infty}_x}\right)\left|\partial_pX(s)\right|\nonumber\\
 & & +C\left\|\partial_xA(s)\right\|_{L^{\infty}_x}\left|\partial_pP(s)\right|.\nonumber
\end{eqnarray}
Hence, since $\dot{A}^i=\partial_sA^i+v_A\cdot\nabla A^i$, $i=1,2,3$ and $\left|D^2v_A(s)\right|\leq C$, the above inequality and (FS$\beta$) yield
\begin{eqnarray}
 \left|\dot{\eta}(s)\right| & \leq & \left(\left\|\partial^2_x\Phi(s)\right\|_{L^{\infty}_x}+\left\|\partial^2_xA(s)\right\|_{L^{\infty}_x}+\left\|\partial_xA(s)\right\|^2_{L^{\infty}_x}\right)\left|\partial_pX(s)\right|\nonumber\\
& & +\left(\left\|\partial_x\Phi(s)\right\|_{L^{\infty}_x}+\left\|\partial_xA(s)\right\|_{L^{\infty}_x}+\left\|\partial_sA(s)\right\|_{L^{\infty}_x}\right)\left|\partial_pP(s)\right|\nonumber\\
& \leq & 2\beta\left(1+s\right)^{-5/2}\left|\partial_pX(s)\right|+\beta\left(1+s\right)^{-3/2}\left|\partial_pP(s)\right|.\nonumber
\end{eqnarray}
Now, by the definition of $\xi(s)$ and $\eta(s)$, we have $\left|\partial_pX(s)\right|\leq\left|\xi(s)\right|+C\left(t-s\right)$ and $\left|\partial_pP(s)\right|\leq C\left(\left|\eta(s)\right|+1\right)$, the latter as a result of $\left|Dv_A^{-1}(s)\right|\leq C$, as it can be easily checked. Then, Gronwall's inequality implies
\begin{eqnarray}
\label{Eta}
   \left|\eta(s)\right| & \leq & \beta Ce^{\beta C}\left(\int^t_s\left(1+\tau\right)^{-5/2}\left|\xi(\tau)\right|d\tau\right.\nonumber\\
   & & \left.+\int^t_s\left[\left(1+\tau\right)^{-5/2}\left(t-\tau\right)+\left(1+\tau\right)^{-3/2}\right]d\tau\right).
\end{eqnarray}
Both (\ref{Xi}) and (\ref{Eta}) then lead to 
\begin{eqnarray}
   \left|\xi(s)\right| & \leq & \beta Ce^{\beta C}\left(\left(t-s\right)+\int^t_s\int^t_{\tau}\left(1+\sigma\right)^{-5/2}\left|\xi(\sigma)\right|d\sigma d\tau\right.\nonumber\\
   & & \left.+\int^t_s\int^t_{\tau}\left[\left(1+\sigma\right)^{-5/2}\left(t-\sigma\right)+\left(1+\sigma\right)^{-3/2}\right]d\sigma d\tau\right),\nonumber\\
   & \leq & \beta Ce^{\beta C}\left(\left(t-s\right)+\int^t_s\int^{\sigma}_s\left(1+\sigma\right)^{-5/2}\left|\xi(\sigma)\right|d\tau d\sigma\right.\nonumber\\
   & & \left.+\int^t_s\int^{\sigma}_s\left[\left(1+\sigma\right)^{-5/2}\left(t-\sigma\right)+\left(1+\sigma\right)^{-3/2}\right] d\tau d\sigma\right),\nonumber\\
   & \leq & \beta Ce^{\beta C}\left(\left(t-s\right)+\int^t_s\left(1+\sigma\right)^{-3/2}\left|\xi(\sigma)\right| d\sigma\right.\nonumber\\
   & & \left.+\int^t_s\left[\left(1+\sigma\right)^{-3/2}\left(t-\sigma\right)+\left(1+\sigma\right)^{-1/2}\right]  d\sigma\right).\nonumber   
\end{eqnarray}
Finally, since the last integral is less than $3\left(t-s\right)$, another use of Gronwall's inequality yields (\ref{First Claim}).\\
\\
\noindent\textbf{Step 2}: For $\beta>0$ small enough, there exists a $C=C(\bar{X}_0,\bar{P}_0)>0$ such that the mapping $X(0,t,x,\cdot):\mathbb{R}^3\rightarrow\mathbb{R}^3$ has Jacobian determinant satisfying
$$  \left|\texttt{det}\partial_pX(0,t,x,p)\right|\geq Ct^3,\quad 0\leq t\leq a,\quad x\in\mathbb{R}^3,\quad p\in\mathbb{R}^3.
$$
For $t=0$ this is obvious. Let $0<t\leq a$. Without loss of generality, we shall assume that $0<\beta\leq 1/2$. Then, by the characteristics and (FS$\beta$) we have
\begin{equation}
\label{Maximum Momentum}
\left|P(t)\right|\leq\bar{P}_0+\beta\int^t_0\left(1+\tau\right)^{-3/2}d\tau\leq\bar{P}_0+1.
\end{equation}
Also, in view of the estimate on the vector potential given in Lemma \ref{Estimate Darwin Potentials}, and recalling that $f_0\in\mathcal{D}$, is not difficult to check that
$$\left\|A(t)\right\|_{L^{\infty}_x}\leq C\bar{X}^2_0\bar{P}^2_0\left(\bar{P}_0+1\right).$$
Denote $g=\left|p-A\right|$. Hence $g\leq C(\bar{X}_0,\bar{P}_0)$ and therefore the relativistic velocity satisfies $\left|v_A\right|\leq\nu<1$, where $\nu$ depends only on $\bar{X}_0$ and $\bar{P}_0$. Now, we have that $Dv_A=\left(1+g^2\right)^{-1/2}\left[\texttt{id}-v_A\otimes v_A\right]$. Then, since by Step 1, $\left|\xi(0)\right|\leq \beta Ce^{\beta C}t$ with $\xi(0)=\partial_pX(0)+tDv_A(t)$, we have for some $\beta>0$ small enough that
\begin{eqnarray}
\label{Gamma}
   \left|\frac{\sqrt{1+g^2}}{t}\partial_pX(0)+\texttt{id}\right| & \equiv & \left|\frac{\sqrt{1+g^2}}{t}\xi(0)+v_A\otimes v_A\right|\nonumber\\
   & \leq & \beta Ce^{\beta C} + \nu \ \equiv \ \gamma \ < \ 1.
\end{eqnarray} 
Therefore, a positive constant $C=C(\bar{X}_0,\bar{P}_0)$ exists such that
$$ \left|\texttt{det}\partial_pX(0)\right|\equiv\frac{t^3}{\left(1+g^2\right)^{3/2}}\left|\texttt{det}\left[\frac{\sqrt{1+g^2}}{t}\partial_pX(0)+\texttt{id}-\texttt{id}\right]\right|\geq Ct^3.
$$ 
\\

\noindent\textbf{Step 3}: For every $0<t\leq a$ and $x\in\mathbb{R}^3$, the mapping $X(0,t,x,\cdot):\mathbb{R}^3\rightarrow\mathbb{R}^3$ is bijective. Indeed, for $p,q\in\mathbb{R}^3$, let
$$ p_{\lambda}=\lambda p+\left(1-\lambda\right)q,\quad g_{\lambda}=g(t,x,p_{\lambda})=\left|p_{\lambda}-A(t,x)\right|,\quad 0\leq\lambda\leq 1.
$$
In view of (\ref{Gamma}) in Step 2, we have
\begin{eqnarray}
  \left|X(0,t,x,p)\right. & - & \left. X(0,t,x,q)\right| \ = \ \left|\int^1_0\partial_pX(0,t,x,p_{\lambda})\left(p-q\right)d\lambda\right|\nonumber\\
   & = & \left|\int^1_0\left[-\texttt{id}+\texttt{id}+\frac{\sqrt{1+g^2_{\lambda}}}{t}\partial_pX(0,t,x,p_{\lambda})\right]\frac{t\left(p-q\right)}{\sqrt{1+g^2_{\lambda}}}d\lambda\right|\nonumber\\
   & \geq & t\left|p-q\right|\int^1_0\frac{d\lambda}{\sqrt{1+g^2_{\lambda}}}-\gamma t\left|p-q\right|\int^1_0\frac{d\lambda}{\sqrt{1+g^2_{\lambda}}}\nonumber\\
   & \geq & \left(1-\gamma\right)\left|p-q\right|t,\nonumber   
\end{eqnarray}
which shows that the mapping is injective. It is also surjective, since the open range $X(0,t,x,\mathbb{R}^3)=\mathbb{R}^3$. If not, there exists a boundary point $x_0$ so that $X(0,t,x,p_n)\rightarrow x_0\notin X(0,t,x,\mathbb{R}^3)$ as $n\rightarrow\infty$, for some $p_n\rightarrow p_0\in\mathbb{R}^3$. By continuity $X(0,t,x,p_0)=x_0$, which is a contradiction, and the assertion follows. \\
\\
\noindent\textbf{Step 4}: Then, Steps 2 and 3 imply that the mapping $X(0,t,x,\cdot):\mathbb{R}^3\rightarrow\mathbb{R}^3$ is a $C^1$-diffeomorphism. In particular, Step 2 implies that for some constant $C=C(\bar{X}_0,\bar{P}_0)>0$, the inverse mapping $X^{-1}(0,t,x,\cdot):\mathbb{R}^3\rightarrow\mathbb{R}^3$ defined by $X\mapsto p(X)$ has Jacobian determinant satisfying
$$  \left|\texttt{det}\partial_pX^{-1}(0,t,x,p(X))\right|\leq Ct^{-3},\quad 0< t\leq a,\quad x\in\mathbb{R}^3.
$$
\\
 We can now deduce the estimates (\ref{First Density Wanted})-(\ref{Second Density Wanted}) for the charge and current densities. Indeed, bearing in mind that $f_0\in\mathcal{D}$, we have
\begin{eqnarray}
  \rho(t,x) & = & \int_{\mathbb{R}^3}f_0(X(0,t,x,p),P(0,t,x,p))dp\nonumber\\
            & = & \int_{\mathbb{R}^3}f_0(X,P(0,t,x,p(X)))\left|\texttt{det}\partial_pX^{-1}(0,t,x,p(X))\right|dX\nonumber\\
            & \leq & Ct^{-3},\nonumber
\end{eqnarray}
where $C=C(\bar{X}_0,\bar{P}_0)>0$. Then, since $\left|j_A\right|\leq\rho$, (\ref{First Density Wanted}) indeed holds.  

 To prove (\ref{Second Density Wanted}) we proceed as follows. In view of (FS$\beta$) with $\beta=1/2$, and recalling that $\left|\partial_xv_A\right|\leq C\left|\partial_xA\right|$ and $f_0\in\mathcal{D}$, we have that
 \begin{eqnarray}
\label{To Long To Fit 1}
   \left\|\partial_x\rho(t)\right\|_{L^{\infty}_x} & \leq & C\left(\bar{P}_0+1\right)^3\left\|\partial_xf(t)\right\|_{L^{\infty}_{x,p}}\\
\label{To Long To Fit 2} 
   \left\|\partial_xj_A(t)\right\|_{L^{\infty}_x} & \leq &
C\left(\bar{P}_0+1\right)^3\left(\left\|\partial_xA(t)\right\|_{L^{\infty}_x}\left\|f_0\right\|_{L^{\infty}_{x,p}}+\left\|\partial_xf(t)\right\|_{L^{\infty}_{x,p}}\right)\nonumber\\
   & \leq & C\left(\bar{P}_0+1\right)^3\left(1+\left\|\partial_xf(t)\right\|_{L^{\infty}_{x,p}}\right),
 \end{eqnarray}
and
\begin{equation}
\label{To Long To Fit 3}
\left|\partial_xf(t,x,p)\right| \leq
 \left|\partial_xP(0,t,x,p)\right|+\left|\partial_xX(0,t,x,p)\right|.
\end{equation}
Hence, the proof will be completed if we provide a uniform bound on the space derivatives of the characteristic curves. Similar to the computations in Step 1, it is not difficult to check that for $0\leq s\leq t\leq a$
$$  \left|\partial_xX(s)\right|\leq 1+C\int^t_s\left(\left|
\partial_xP(\tau)\right|+\left\|\partial_xA(\tau)\right\|_{L^{\infty}_x}\left|\partial_xX(\tau)\right|\right)d\tau,
$$   
and by (FS$\beta$) and Gronwall's lemma 
\begin{equation}
\label{Estimate Partial Characteristic X}
\left|\partial_xX(s)\right|\leq C\left(1+\int^t_s\left|
\partial_xP(\tau)\right|d\tau\right).
\end{equation}
Also,
\begin{eqnarray}
\label{Estimate Partial Characteristic P}
  \left|\partial_xP(s)\right| & \leq & C\int^t_s\left(\left\|\partial^2_x\Phi(s)\right\|_{L^{\infty}_x}+\left\|\partial^2_xA(s)\right\|_{L^{\infty}_x}+\left\|\partial_xA(s)\right\|^2_{L^{\infty}_x}\right)\left|\partial_xX(\tau)\right|d\tau\nonumber\\
  & & +C\int^t_s\left\|\partial_xA(\tau)\right\|_{L^{\infty}_x}\left|\partial_xP(\tau)\right|d\tau\nonumber\\
  & \leq & C\int^t_s\left(1+\tau\right)^{-5/2}\left|
\partial_xX(\tau)\right|d\tau.
\end{eqnarray}
Therefore, both (\ref{Estimate Partial Characteristic X}) and (\ref{Estimate Partial Characteristic P}) yield
\begin{eqnarray}
 \left|\partial_xX(s)\right| & \leq & C+C\int^t_s\int^t_{\tau}\left(1+\sigma\right)^{-5/2}\left|
\partial_xX(\sigma)\right|d\sigma d\tau\nonumber\\
& \leq & C+C\int^t_s\int^{\sigma}_s\left(1+\sigma\right)^{-5/2}\left|
\partial_xX(\sigma)\right|d\tau d\sigma\nonumber\\
& \leq & C+C\int^t_s\left(1+\sigma\right)^{-3/2}\left|
\partial_xX(\sigma)\right|d\sigma.\nonumber
\end{eqnarray}
Gronwall's lemma then provides a uniform bound on $\left|\partial_xX(s)\right|$, which in turn produces a uniform bound on $\left|\partial_xP(s)\right|$ via (\ref{Estimate Partial Characteristic P}). As a consequence 
$$ \left|\partial_xX(0,t,x,p)\right|+\left|\partial_xP(0,t,x,p)\right|\leq C,\hspace{.5cm} 0\leq t\leq a,\hspace{.1cm}x\in\mathbb{R}^3,\hspace{.1cm} p\in\mathbb{R}^3, 
$$
which implies (\ref{Second Density Wanted}) via (\ref{To Long To Fit 1})-(\ref{To Long To Fit 3}). This concludes the proof of the lemma.
\end{proof}

%%%%%%%%%%%%%%%%%%%%%%%%

\begin{lemma}
\label{Lemma Estimates second}
Under the assumptions of Lemma \ref{Lemma Estimates}, we have:
  \begin{eqnarray}
  \label{No Potentials Wanted}
    \left\|\partial_t\partial_x A(t)\right\|_{L^2_x}+\left\|\rho^{1/2}(t)\partial_tA(t)\right\|_{L^2_x} & \leq & Ct^{-3/2},\\
   \label{First Potentials Wanted}
   \left\|\partial_x\Phi(t)\right\|_{L^{\infty}_x}+\left\|\partial_xA(t)\right\|_{L^{\infty}_x}  + \left\|\partial_tA(t)\right\|_{L^{\infty}_x} & \leq & Ct^{-2},\\
   \label{Second Potentials Wanted}
   \left\|\partial^2_x\Phi(t)\right\|_{L^{\infty}_x}+\left\|\partial^2_xA(t)\right\|_{L^{\infty}_x} & \leq & Ct^{-3}\ln(1+t),
  \end{eqnarray}
  for all $0\leq t\leq a$.
\end{lemma}

\begin{proof}
 By virtue of Lemma \ref{Estimate Darwin Potentials}, the following estimates hold
 \begin{eqnarray}
  \left\|\partial_x\Phi(t)\right\|_{L^{\infty}_x}+\left\|\partial_xA(t)\right\|_{L^{\infty}_x} & \leq & C\left\|f_0\right\|^{1/3}_{L^{1}_{x,p}}\left\|\rho(t)\right\|^{2/3}_{L^{\infty}_x},\nonumber\\
 \left\|\partial^2_x\Phi(t)\right\|_{L^{\infty}_x}+\left\|\partial^2_xA(t)\right\|_{L^{\infty}_x} & \leq &  Ct^{-3}\left[\left\|f_0\right\|^{1/3}_{L^{1}_{x,p}}+\left\|\partial_x\rho(t)\right\|_{L^{\infty}_x}+\left\|\partial_xj_A(t)\right\|_{L^{\infty}_x}\right.\nonumber\\
  && \left.+\left(1+\ln t^4\right)t^3\left\|\rho(t)\right\|_{L^{\infty}_x}\right],\quad  t>1,\nonumber
 \end{eqnarray}
where the latter is a consequence of setting $R=t$ and $h=t^{-3}\leq R$ in the cited lemma. Then using Lemma \ref{Lemma Estimates}, we conclude (\ref{Second Potentials Wanted}) and 
\[ \left\|\partial_x\Phi(t)\right\|_{L^{\infty}_x}+\left\|\partial_xA(t)\right\|_{L^{\infty}_x} \leq  Ct^{-2}.\]
Moreover, the estimate $\left\|\partial_tA(t)\right\|_{L^{\infty}_x}  \leq  Ct^{-2}$ follows from (\ref{No Potentials Wanted}) and Lemmas \ref{LInfinity Estimate Time Derivative A} and \ref{Lemma Estimates}, as we show next. From Lemma \ref{LInfinity Estimate Time Derivative A} and the fact that
\[\|\rho(t)\|_{L^1_x} =\|f(t)\|_{L^1_{x,p}} = \|f_0\|_{L^1_{x,p}} \leq 1, \]
we have
\[\left\|\partial_tA(t)\right\|_{L^{\infty}_x} \leq C \left[\|\rho(t)\|^{2/3}_{L^\infty_x} (1+\|\rho(t)\|^{1/3}_{L^\infty_x}) + \|\rho(t)\|^{1/3}_{L^\infty_x}  \|\rho^{1/2}(t) \partial_t A(t)\|_{L^2_x} \right]. \]
Then using (\ref{First Density Wanted}) and (\ref{No Potentials Wanted}), we deduce that for $t>1$,
\begin{eqnarray*}
\left\|\partial_tA(t)\right\|_{L^{\infty}_x} &\leq& C \left[t^{-2} + t^{-3} + t^{-1}\|\rho^{1/2}(t) \partial_t A(t)\|_{L^2_x}\right] \\
& \leq & C \left[t^{-2} + t^{-3} + t^{-5/2}\right] \leq C t^{-2}
\end{eqnarray*} 
as desired. It remains to prove (\ref{No Potentials Wanted}). This estimate can be obtained by following the lines of the proof of Lemma \ref{L2 Estimate Time Derivative A} while keeping track of the time dependence in the estimates. Indeed, from (\ref{LHS and RHS}) and (\ref{inequality on A}), we have

\begin{eqnarray*}
\int_{\mathbb{R}^3}\left|\partial_t\partial_xA(t)\right|^2 dx & + & \int_{\mathbb{R}^3}\int_{\mathbb{R}^3}\frac{f (1-|v_A|^2)}{\sqrt{1+|g_A(t)|^2}}\left|\partial_tA(t)\right|^2\\
& \leq & \int_{\mathbb{R}^3}\int_{\mathbb{R}^3}f\partial_tA^i\left(v_A\cdot\nabla_xv^i_A\right)+\int_{\mathbb{R}^3}\int_{\mathbb{R}^3}fv^i_A\left(v_A\cdot\nabla_x\partial_tA^i\right)\\
   & & + \int_{\mathbb{R}^3}\int_{\mathbb{R}^3}f\partial_tA^i\left(K\cdot\nabla_pv^i_A\right),
\end{eqnarray*}
where $g_A(t)=p-A(t)$ and $K=-\nabla\Phi + v_A^i\nabla A^i$. Since $|v_A| \leq 1$, $|\partial_xv_A| \leq C|\partial_xA|$, $|\partial_pv_A| \leq C$ and $1-|v_A|^2=(1+|g_A|^2)^{-1}$, we have
\begin{eqnarray}
\label{LHS and RHS prime}
\int_{\mathbb{R}^3}\left|\partial_t\partial_xA(t)\right|^2 dx & + & \int_{\mathbb{R}^3}\int_{\mathbb{R}^3}\frac{f}{(1+|g_A(t)|^2)^{3/2}}\left|\partial_tA(t)\right|^2 dx dp
\nonumber\\
& \leq & C \left[ \int_{\mathbb{R}^3} \rho(t) |\partial_tA(t)| |\partial_x A(t)| dx  +\int_{\mathbb{R}^3} \rho(t) |\partial_t\partial_x A(t)| dx \right] \nonumber\\
   & & +\; C \int_{\mathbb{R}^3}\rho(t) \left( |\partial_x\Phi(t)| + |\partial_xA(t)|\right) |\partial_tA(t)| dx.
   \end{eqnarray}
We now proceed to estimate both sides of the inequality (\ref{LHS and RHS prime}). By Lemma \ref{Estimate Darwin Potentials}, the right-hand side term gives
\begin{eqnarray*}
\texttt{RHS} &\leq& C\left[\left( \|\partial_x\Phi(t)\|_{L^\infty_x} + \|\partial_xA(t)\|_{L^\infty_x}\right) \int_{\R^3}\rho(t)|\partial_tA(t)| + \int_{\R^3}\rho(t)|\partial_t\partial_xA(t)| \right] \\
& \leq& C\left[ \|\rho(t)\|^{1/3}_{L^1_x} \|\rho(t)\|^{2/3}_{L^\infty_x}\int_{\R^3} \rho(t)|\partial_t A(t)|  + \int_{\R^3}\rho(t)|\partial_t\partial_xA(t)| \right].
\end{eqnarray*}
By Cauchy-Schwarz inequality, 
\[\int_{\R^3}\rho(t)|\partial_t A(t)|   \leq  \|\rho^{1/2}(t)\|_{L^2_x} \|\rho^{1/2}(t)\partial_tA(t)\|_{L^2_x} = \|\rho(t)\|_{L^1_x}^{1/2} \|\rho^{1/2}(t)\partial_tA(t)\|_{L^2_x} \]
and
\[ \int_{\R^3}\rho(t)|\partial_t\partial_xA(t)| \leq \|\rho(t)\|_{L^2_x} \|\partial_t\partial_x A(t)\|_{L^2_x}  \leq \|\rho(t)\|^{1/2}_{L^1_x} \|\rho(t)\|^{1/2}_{L^\infty_x}  \|\partial_t\partial_x A(t)\|_{L^2_x}. \]
Then using that $\|\rho(t)\|_{L^1_x} \leq 1$, we have
 \begin{eqnarray*}
\texttt{RHS} &\leq& C\, \|\rho(t)\|^{5/6}_{L^1_x} \|\rho(t)\|^{2/3}_{L^\infty_x} \|\rho^{1/2}(t)\partial_tA(t)\|_{L^2_x}\\
& &  \quad + \,  C\, \|\rho(t)\|^{1/2}_{L^1_x} \|\rho(t)\|^{1/2}_{L^\infty_x}  \|\partial_t\partial_x A(t)\|_{L^2_x} \\
&\leq& C\left[ \|\rho(t)\|^{2/3}_{L^\infty_x} + \|\rho(t)\|^{1/2}_{L^\infty_x}\right] \left [ \|\rho^{1/2}(t)\partial_tA(t)\|_{L^2_x} +  \|\partial_t\partial_x A(t)\|_{L^2_x} \right].
 \end{eqnarray*}
Finally, we deduce from (\ref{First Density Wanted}) in Lemma \ref{Lemma Estimates} that for $t>1$,
\begin{eqnarray}\label{estimate rhs term}
\texttt{RHS} &\leq& C(t^{-2} + t^{-3/2}) \left [ \|\rho^{1/2}(t)\partial_tA(t)\|_{L^2_x} +  \|\partial_t\partial_x A(t)\|_{L^2_x} \right] \nonumber\\
&\leq& Ct^{-3/2} \left [ \|\rho^{1/2}(t)\partial_tA(t)\|_{L^2_x} +  \|\partial_t\partial_x A(t)\|_{L^2_x} \right].
\end{eqnarray}
 As for the term in the left-hand side of (\ref{LHS and RHS prime}), we proceed as follows. First, we notice that $|g_A(t,x,p)| \leq 2\left(|p|^2 + \|A(t)\|^2_{L^\infty_x}\right)$, and by Lemma \ref{Estimate Darwin Potentials} with $\|\rho(t)\|_{L^1_x} \leq 1$ we have $\|A(t)\|_{L^\infty_x} \leq C\|\rho(t)\|^{1/3}_{L^\infty_x}$. Then, by (\ref{First Density Wanted}) in Lemma \ref{Lemma Estimates},
 \[ |g_A(t,x,p)| \leq C\left( |p|^2 + t^{-2} \right). \]
 From the characteristic equation (\ref{charact P}), 
 \[ |P(t)| \leq |P(0)| + \int_0^t \left(\|\partial_x\Phi(s)\|_{L^\infty_x} + \|\partial_xA(s)\|_{L^\infty_x} \right) ds, \]
 and by Lemma \ref{Estimate Darwin Potentials} and $\|\rho(t)\|_{L^1_x} \leq 1$,
 \[ \|\partial_x\Phi(s)\|_{L^\infty_x} + \|\partial_xA(s)\|_{L^\infty_x} \leq C\|\rho(s)\|^{1/3}_{L^1_x} \|\rho(s)\|^{2/3}_{L^\infty_x} \leq C \|\rho(s)\|^{2/3}_{L^\infty_x}.\]
 Then, using  that $\|\rho(s)\|_{L^\infty_x} \leq Cs^{-3}$ on $s>1$ by (\ref{First Density Wanted}) in Lemma \ref{Lemma Estimates}, and that $\|\rho(s)\|_{L^\infty_x}$ is uniformly bounded on $s\in [0,1]$ (by Theorem \ref{Local Solutions Theorem} and Lemma \ref{lemma_noname}), we have that $\|\rho(s)\|_{L^\infty_x} \leq C(1+s)^{-3}$ for $0\leq s\leq t$ with $t>1$, so that
 \[|P(t)| \leq |P(0)| + C\int_0^t  (1+s)^{-2} ds \leq |P(0)| + C \equiv C.  \]
 Therefore, for all $t>1$, we obtain that
 \[|g_A(t,x,p)| \leq C(1+t^{-2}) \leq 2C.\]
 We then deduce that the left hand side term of (\ref{LHS and RHS prime}) can be estimated as:
 \begin{eqnarray}\label{estimate lhs term}
 \texttt{LHS} &\geq& \int_{\R^3}|\partial_t\partial_x A(t)|^2 dx + C\int_{\R^3}\int_{\R^3} f|\partial_tA(t)|^2 dx dp \nonumber\\
 &=&  \|\partial_t\partial_xA(t)\|^2_{L^2_x} + C \|\rho^{1/2}(t)\partial_tA(t)\|^2_{L^2_x} \nonumber\\
 &\geq& C\left[ \|\partial_t\partial_xA(t)\|^2_{L^2_x} +  \|\rho^{1/2}(t)\partial_tA(t)\|^2_{L^2_x}\right].
 \end{eqnarray}
 Hence, if we combine (\ref{LHS and RHS prime}) - (\ref{estimate lhs term}), we deduce that 
 \begin{eqnarray*}
 \lefteqn{ \|\partial_t\partial_xA(t)\|^2_{L^2_x} +  \|\rho^{1/2}(t)\partial_tA(t)\|^2_{L^2_x}}\\
 & &  \leq Ct^{-3/2} \left [  \|\partial_t\partial_x A(t)\|_{L^2_x} + \|\rho^{1/2}(t)\partial_tA(t)\|_{L^2_x} \right].
 \end{eqnarray*}
 The above inequality together with $(a+b)^2 \leq 2(a^2+b^2)$ conclude (\ref{No Potentials Wanted}). 
 \end{proof}

%%%%%%%%%%%%%%%%%%%%%%%%%

\paragraph{\textit{\textbf{Proof of Theorem \ref{Global Solutions Theorem}}}} 
By virtue of the previous lemmas, the proof is almost identical to the proof of \cite[Theorem 4.1]{Rein} for the Vlasov-Poisson system. 

Indeed, let $\beta,\delta>0$ and $C=C(\bar{X}_0,\bar{P}_0)>0$ be suitable for Lemma \ref{Lemma Estimates} to hold. Fix $T_0>1$ such that for all $t\geq T_0$ 
\begin{equation}
\label{Suitable Estimate}
Ct^{-2}\leq\frac{\beta}{2}\left(1+t\right)^{-3/2},\hspace{.5cm} C\left(1+\ln t\right)t^{-3}\leq\frac{\beta}{2}\left(1+t\right)^{-5/2}.
\end{equation}

Now, by letting $\delta>0$ be smaller if necessary, Lemma \ref{lemma_noname} implies that the Cauchy datum $f_0\in\mathcal{D}$ with $\left\|f_0\right\|_{L^{\infty}_{x,p}}\leq\delta$ yields a classical solution $f$ of the RVD system on the maximal existence interval $[0,T[$ with $T>T_0$, and 
\begin{eqnarray} 
\left\|\partial_tA(t)\right\|_{L^{\infty}_x}+\left\|\partial_xA(t)\right\|_{L^{\infty}_x}+\left\|\partial_x\Phi(t)\right\|_{L^{\infty}_x}\hspace{3.cm}\nonumber\\
+\left\|\partial^2_xA(t)\right\|_{L^{\infty}_x}+\left\|\partial^2_x\Phi(t)\right\|_{L^{\infty}_x}<\frac{\beta}{2}\left(1+T_0\right)^{-5/2},\nonumber
\end{eqnarray}
for all $0\leq t\leq T_0$. Hence, $f$ satisfies the free streaming condition (FS$\beta$) on $[0,T_0]$. In fact, the continuity of the left-hand side of the above inequality implies that there exists a maximal $T_0<T_1\leq T$ such that $f$ satisfies (FS$\beta$) on $[0,T_1[$. Therefore, Lemma \ref{Lemma Estimates second} and (\ref{Suitable Estimate}) imply that for all $T_0\leq t<T_1$
\begin{eqnarray}
 \left\|\partial_tA(t)\right\|_{L^{\infty}_x} + \left\|\partial_x\Phi(t)\right\|_{L^{\infty}_x}+\left\|\partial_xA(t)\right\|_{L^{\infty}_x}  & \leq & Ct^{-2}\hspace{.1cm}\leq \hspace{.1cm}\frac{\beta}{2}(1+t)^{-3/2},\nonumber\\
\left\|\partial^2_x\Phi(t)\right\|_{L^{\infty}_x}+\left\|\partial^2_xA(t)\right\|_{L^{\infty}_x} & \leq & C\left(1+\ln t\right)t^{-3}\hspace{.1cm}\leq\hspace{.1cm}\frac{\beta}{2}(1+t)^{-5/2}.\nonumber
\end{eqnarray}
Then, a continuation argument yields $T_1=T$, and by (\ref{Maximum Momentum}), we deduce
   $$\sup\left\{\left|p\right|:\exists0\leq t<T,x\in\mathbb{R}^3:f(t,x,p)\neq0\right\}\leq\bar{P}_0+1.$$
Therefore the continuation criterion in Theorem \ref{Local Solutions Theorem} implies that $T=\infty$, and thus the solution $f$ is global in time. The proof of Theorem \ref{Global Solutions Theorem} is complete.\qed

%%%%%%%%%%%%%%%%%%%%%%%%%%%%%%%%%%%%%%%%%%%%%%%%%%%%%%%%%%
%%%%%%%%%%%%%%%%%%%%%%%%%%%%%%%%%%%%%%%%%%%%%%%%%%%%%%%%%%

\section*{Appendix}
\label{C}

\noindent $$\hbox{For}\quad v_A=\frac{p-A}{\sqrt{1+\left|p-A\right|^2}},\quad \hbox{set}\quad 
Dv_A=\frac{\texttt{id}-v_A\otimes v_A}{\sqrt{1+\left|p-A\right|^2}}=\partial_pv_A.$$
Then $\partial_xv_A=-Dv_A\partial_xA$ and $\partial_tv_A=-Dv_A\partial_tA$. Clearly, $\left|Dv_A\right|\leq C$, and so $\left|\partial_xv_A\right|\leq C\left|\partial_xA\right|$ and $\left|\partial_tv_A\right|\leq C\left|\partial_tA\right|$. Also, $\left|\partial^2_pv_A\right|\leq C$; $\left|\partial_x\partial_pv_A\right|\leq C\left|\partial_xA\right|$; $\left|\partial^2_xv_A\right|\leq C\left(\left|\partial^2_xA\right|+\left|\partial_xA\right|^2\right)$; $\left|\partial^2_tv_A\right| \leq C\left(\left|\partial^2_tA\right|+\left|\partial_tA\right|^2\right)$;$\left|\partial_t\partial_pv_A\right|\leq  C\left|\partial_tA\right|$ and finally
$\left|\partial_t\partial_xv_A\right| \leq C\left(\left|\partial_t\partial_xA\right|+\left|\partial_tA\right|\left|\partial_xA\right|\right)$.

%%%%%%%%%%%%%%%%%%%%%%%%%%%%%%%% Final %%%%%%%%%%%%%%%%%%%%%%%%%

\end{document}